\documentclass[10pt]{article}
\usepackage{amsfonts}
\usepackage[top=0.9in,bottom=1.1in,left=1.05in,right=1.05in]{geometry}
\usepackage{amsmath,amssymb}
\usepackage{amsthm,bm}
\usepackage{threeparttable}
\usepackage{latexsym}
\usepackage{mathrsfs,dsfont}
\usepackage{cancel}
\usepackage{comment}
\usepackage{float}
\usepackage{graphicx}
\usepackage{epstopdf}
\usepackage{color}
\usepackage{enumerate}
\usepackage{algorithm,algorithmic}
\usepackage{mathtools}
\usepackage{tabu}
\usepackage{subfigure}
\usepackage{hyperref}
\usepackage[font=small,labelfont=bf]{caption}

\DeclareGraphicsExtensions{.eps}
\numberwithin{equation}{section}

\newcommand{\mc}[1]{\mathcal{#1}}

\newcommand{\PP}{\mathbb{P}}
\newcommand{\RR}{\mathbb{R}}

\newcommand{\QQ}{\mathbb{Q}}

\newcommand{\half}{\frac{1}{2}}

\newcommand{\abs}[1]{\left|#1\right|}

\newtheorem{theo}{Theorem}[section]
\newtheorem{lem}[theo]{Lemma}

\newtheorem{rem}[theo]{Remark}
\newtheorem{prop}[theo]{Proposition}
\newtheorem{assump}[theo]{Assumption}

\newtheorem{defi}[theo]{Definition}

\begin{document}

\title{\vspace{-20pt}$N$-player and mean-field games in It\^{o}-diffusion
markets with competitive or homophilous interaction }
\author{Ruimeng Hu\thanks{%
Department of Mathematics and Department of Statistics \& Applied
Probability, University of California, Santa Barbara, CA 93106-3080, \emph{%
rhu@ucsb.edu}. } \and Thaleia Zariphopoulou\thanks{%
Departments of Mathematics and IROM, The University of Texas at Austin,
Austin, USA, and the Oxford-Man Institute, University of Oxford, Oxford, UK, 
\emph{zariphop@math.utexas.edu}.}}
\date{\today}
\maketitle

\begin{abstract}
In It\^{o}-diffusion environments, we introduce and analyze $N$-player and common-noise mean-field games in the context of optimal portfolio choice in a common market. The players invest in a finite horizon and also interact, driven either by competition or homophily. We study an incomplete market model in which the players have constant individual risk tolerance coefficients (CARA utilities). We also consider the general case of random individual risk tolerances and analyze the related games in a complete market setting. This randomness makes the problem substantially more complex as it leads to ($N$ or a continuum of) auxiliary ``individual'' It\^{o}-diffusion markets. For all cases, we derive explicit or closed-form solutions for the equilibrium stochastic processes, the optimal state processes, and the values of the games. 
\end{abstract}

\section{Introduction}

In It\^{o}-diffusion environments, we introduce $N$-player and common-noise mean-field games (MFGs) in the context of optimal portfolio choice in a common market. We build on the framework and notions of \cite{LaZa:17} (see, also, \cite{lacker2020many}) but allow for a more general market model (beyond the log-normal case) and, also, consider more complex risk preferences. 

The paper consists of two parts. In the first part, we consider a common incomplete market and players with individual exponential utilities (CARA) who invest while interacting with each other, driven either by competition or homophily. We derive the equilibrium policies, which turn out to be state (wealth)-independent stochastic processes. Their forms depend on the market dynamics, the risk tolerance coefficients, and the underlying minimal martingale measure. We also derive the optimal wealth and the values of both the $N$-player and the mean-field games, and discuss the competitive and homophilous cases.

In the second part, we assume that the common It\^{o}-diffusion market is complete, but we generalize the model in the direction of risk preferences, allowing the risk tolerance coefficients to be random variables. For such preferences, we first analyze the single-player problem, which is interesting in its own right. Among others, we show that the randomness of the utility ``distorts'' the original market by inducing a ``personalized'' risk premium process. This effect is more pronounced in the $N$-player game where the common market is now replaced by ``personalized'' markets whose stochastic risk premia depend on the individual risk tolerances. As a result, the tractability coming from the common market assumption is lost. In the MFG setting, these auxiliary individual markets are randomly selected (depending on the type vector) and aggregate to a common market with a modified risk premium process. We characterize the optimal policies, optimal wealth processes, and game values, building on the aforementioned single-player problem.

To our knowledge, $N$-player games and MFGs in It\^{o}-diffusion market settings have not been considered before except in preprint \cite{fu2020mean}. Therein, the authors used the same asset specialization framework and same CARA preferences as in \cite{LaZa:17} but allowed for It\^{o}-diffusion price dynamics. They studied the problem using a forward-backward stochastic differential equation (FBSDE) approach. In our work, we have different model settings regarding both the measurability of the coefficients of the It\^{o}-diffusion price processes and the individual risk tolerance inputs. We also solve the problems using a different approach, based on the analysis of portfolio optimization problems of exponential utilities in semi-martingale markets.  

The theory of mean-field games was introduced by Lasry and Lions \cite{lasry2007mean}, who developed the fundamental elements of the mathematical theory and, independently, by Huang, Malham\'{e} and Caines who considered a particular class \cite{huang2006large}. Since then, the area has grown rapidly both in terms of theory and applications. Listing precise references is beyond the scope of this paper.

Our work contributes to $N$-player games and MFG in It\^{o}-diffusion settings for models with controlled processes whose dynamics depend linearly on the controls and are state-independent, and, furthermore, the controls appear in both the drift and the diffusion parts. Such models are predominant in asset pricing and in optimal portfolio and consumption choice. In the context of the general MFG theory, the models considered herein are restrictive. On the other hand, their structure allows us to produce explicit/closed-form solutions for It\^{o}-diffusion environments.

The paper is organized as follows. In Section~\ref{sec:incomplete}, we study the incomplete market case for both the $N$-player game and the MFG, and for CARA utilities. In Section~\ref{sec:complete}, we focus on the complete market case but allow for random risk tolerance coefficients. In analogy to Section~\ref{sec:incomplete}, we analyze both the $N$-player game and the MFG. We conclude in  Section~\ref{sec:conclude}.

\section{Incomplete It\^{o}-diffusion common market and CARA utilities}\label{sec:incomplete}

We consider an incomplete It\^{o}-diffusion market, in which we introduce an $N$-player and a mean-field game for players who invest in a finite horizon while interacting among them, driven either by competition or homophily. We assume that the players (either at the finite or the continuum setting) have individual constant risk tolerance coefficients. For both the $N$-player and the MFG, we derive in closed form the optimal policies, optimal controlled processes, and the game values. The analysis uses the underlying minimal martingale measure, related martingales, and their decomposition.

\subsection{The $N$-player game}\label{sec_Nagent}

Consider a probability space $(\Omega ,\mathcal{F},\mathbb{P})$ supporting two Brownian motions $(W_{t},W_{t}^{Y})_{t\in [0,T]},$ $T<\infty ,$ imperfectly correlated with the correlation coefficient $\rho \in \left( -1,1\right) $. We denote by $(\mc{F}_t)_{t \in [0,T]}$ the natural filtration generated by both $W$ and $W^Y$, and by $(\mc{G}_t)_{t \in [0,T]}$ the one generated only by $W^Y$. We then let $(\mu _{t}) _{t\in [ 0,T]}$ and $(\sigma
_{t})_{t\in \left[ 0,t\right] }$ be $\mathcal{G}_{t}$-adapted
processes, with $0<c \leq \sigma _{t}\leq C$ and $\abs{\mu _{t}}\leq C$, $t\in \left[
0,T\right]$,  for some (possibly deterministic) constants $c$ and $C$. 

The financial market consists of a riskless bond (taken to be the numeraire
and with zero interest rate) and a stock whose price process $(S_{t}) _{t\in \left[ 0,T\right] }$ satisfies 
\begin{equation}
\,dS_{t}=\mu _{t}S_{t}\,dt+\sigma _{t}S_{t}\,dW_{t},  \quad S_0 = s_0 \in \RR^+. \label{stock-1}
\end{equation}%

In this market, $N$ players, indexed by $i\in \mathcal{I},$ $\mathcal{I}%
=\left \{ 1, 2, \ldots, N\right \}$, have a common investment horizon $\left[ 0,T\right] $
and trade between the two accounts. Each player, say player $i,$ uses a
self-financing strategy $(\pi _{t}^{i})_{t\in \lbrack 0,T]}$, representing
(discounted by the numeraire) the amount invested in the stock. Then, her
wealth $( X_{t}^{i}) _{t\in [0,T] }$ satisfies 
\begin{equation}
\,dX_{t}^{i}=\pi _{t}^{i}\left( \mu _{t}\,dt+\sigma _{t}\,dW_{t}\right)
,\quad \  \ X_{0}^{i}=x_{i}\in \mathbb{R},  \label{X-i}
\end{equation}%
with $\pi ^{i}$ being an admissible
policy, belonging to
\begin{equation}
\mathcal{A}=\left \{ \pi :\text{self-financing, }\mathcal{F}\text{%
-progressively measurable and }E_{\mathbb{P}}\left[ \int_{0}^{T}\sigma_s^2\pi
_{s}^{2}\,ds\right] <\infty \right \} .  \label{admissible-set-part1}
\end{equation}

As in \cite{LaZa:17} (see also \cite{basak2015competition, espinosa2015optimal,
huang2016mean,kraft2020dynamic, lacker2020many, whitmeyer2019relative}),
players optimize their expected terminal utility but are, also, concerned
with the performance of their peers. For an arbitrary but \textit{fixed} 
policy $( \pi _{1}, \ldots, \pi_{i-1}, \pi_{i+1}, \ldots,\pi _{N})$, 
player $i$, $i\in \mathcal{I}$, seeks to optimize 
\begin{equation}
V^{i}\left( x_{1}, \ldots, x_{i}, \ldots, x_{N}\right) =\sup_{\pi ^{i}\in \mathcal{A}%
}E_{\mathbb{P}}\left[ \left. -\exp \left( -\frac{1}{\delta _{i}}\left(
X_{T}^{i}-c_{i}C_{T}\right) \right) \right \vert
X_{0}^{1}=x_{1}, \ldots, X_{0}^{i}=x_{i}, \ldots, X_{0}^{N}=x_{N}\right],
\label{i-utility}
\end{equation}%
where 
\begin{equation}
C_{T}:=\frac{1}{N}\sum_{j=1}^{N}X_{T}^{j}  \label{C-process}
\end{equation}%
averages all players' terminal wealth, with $X_{T}^{j}$, $j=1, \ldots, N$, given by \eqref{X-i}.

The parameter $\delta _{i}>0$ is the individual (absolute) risk tolerance
while the constant $c_{i}\in (-\infty ,1]$ models the individual interaction
weight towards the average wealth of all players. If $c_{i}>0$,
the above criterion models \textit{competition} while when $c_{i}<0$ it
models \textit{homophilous} interactions (see, for example, \cite{leng2020learning}). The
optimization criterion \eqref{i-utility} can be, then, viewed as a
stochastic game among the $N$ players, where the notion of
optimality is being considered in the context of \textit{a Nash equilibrium}, 
stated below (see, for example, \cite{carmona2016lectures}).

\begin{defi}\label{def:Nash}
A strategy $(\pi _{t}^{\ast}) _{t\in \left[
0,T\right] }=(\pi _{t}^{1,\ast },\ldots ,\pi _{t}^{N,\ast })_{t\in [0,T] }\in \mathcal{A}^{\otimes N}$ is called a Nash equilibrium
if, for each $i\in \mathcal{I}$ and $\pi ^{i}\in \mathcal{A},$ 
\begin{multline}\label{Nash}
E_{\mathbb{P}}\left[ \left. -\exp \left( -\frac{1}{\delta _{i}}\left(
X_{T}^{i,\ast }-c_{i}C_{T}^{\ast }\right) \right) \right \vert
X_{0}^{1}=x_{1}, \ldots, X_{0}^{i}=x_{i}, \ldots, X_{0}^{N}=x_{N}\right]  \\
\geq E_{\mathbb{P}}\left[ \left. -\exp \left( -\frac{1}{\delta _{i}}\left(
X_{T}^{i}-c_{i}C_{T}^{i,\ast }\right) \right) \right \vert
X_{0}^{1}=x_{1}, \ldots, X_{0}^{i}=x_{i}, \ldots, X_{0}^{N}=x_{N}\right]
\end{multline}
with 
\begin{equation*}
C_{T}^{\ast }:=\frac{1}{N}\sum_{j=1}^{N}X_{T}^{j,\ast }\text{ \  \ and \  \ }%
C_{T}^{i,\ast }:=\frac{1}{N}\left( \sum_{j=1,j\neq i}^{N}X_{T}^{j,\ast
}+X_{T}^{i}\right) ,\text{\ }
\end{equation*}%
where $X_{T}^{j,\ast }$, $j\in \mathcal{I}$, solve \eqref{X-i} with $\pi^{j,\ast }$ being used. 

\end{defi}

In this incomplete market, we recall the associated \textit{minimal
martingale} measure $\mathbb{Q}^{MM}$,
defined on $\mathcal{F}_{T},$ with 
\begin{equation}\label{def_MM}
\frac{\,d\mathbb{Q}^{MM}}{\,d\mathbb{P}}=\exp \left( -\frac{1}{2}%
\int_{0}^{T}\lambda _{s}^{2}\,ds-\int_{0}^{T}\lambda _{s}\,dW_{s}\right) ,
\end{equation}%
where $\lambda _{t}:=\frac{\mu _{t}}{\sigma _{t}}$, $t\in [0,T]$, is the Sharpe ratio process (see, among others, \cite{FoSc:91}). By the assumptions on the model
coefficients, we have that, for $t\in \left[ 0,T\right] $, $\lambda _{t}\in 
\mathcal{G}_{t} $ and
\begin{equation}
\abs{\lambda _{t}}\leq K,  \label{lamda-bounds}
\end{equation}%
for some (possibly deterministic) constant $K$. We also consider the
processes $(\widetilde{W}_{t}) _{t\in \left[ 0,T\right] }$ and $(\widetilde{W}_{t}^{Y})_{t\in [ 0,T] }$ with $%
\widetilde{W}_{t}=W_{t}+\int_{0}^{t}\lambda _{s}\,ds$ and $\widetilde{W}%
_{t}^{Y}=W_{t}^{Y}+\rho \int_{0}^{t}\lambda _{s}\,ds,$ which are standard
Brownian motions under $\mathbb{Q}^{MM}$ with $\widetilde{W}_{t}\in \mathcal{F}_{t}$ and $\widetilde{%
W}_{t}^{Y}\in \mathcal{G}_{t}.$

Next, we introduce the $\mathbb{Q}^{MM}$-martingale $(M_{t})
_{t\in \left[ 0,T\right] }$,  
\begin{equation}
M_{t}:=E_{\mathbb{Q}^{MM}}\left[ \left. e^{-\frac{1}{2}(1-\rho
^{2})\int_{0}^{T}\lambda _{s}^{2}\,ds}\right \vert \mathcal{G}_{t}\right] .
\label{martingale}
\end{equation}%
From \eqref{lamda-bounds} and the martingale representation theorem, there exists a $\mathcal{G}_{t}$-adapted process $\xi\in 
\mathcal{L}^{2}\left(\mathbb{P}\right) $ such that 
\begin{equation}
\,dM_{t}=\xi _{t}M_{t}\,d\widetilde{W}_{t}^{Y}=\xi _{t}M_{t}\left( \rho \,d%
\widetilde{W}_{t}+\sqrt{1-\rho ^{2}}\,dW_{t}^{\perp }\right) ,
\label{def_xi}
\end{equation}%
where $W_{t}^{\perp }$ is a standard Brownian motion independent of $W_{t}$
appearing in the decomposition $W_{t}^{Y}=\rho W_{t}+\sqrt{1-\rho ^{2}}%
W_{t}^{\perp }$.

\medskip

In the absence of interaction among the players ($c_{i}\equiv 0,$ $i\in 
\mathcal{I}),$ the optimization problem \eqref{i-utility} has been analyzed
by various authors (see, among others, \cite{RoEl:00,SiZa:05}). We recall
its solution which will be frequently used herein.

\begin{lem}[no interaction]\label{lem:nocompetition}
Consider the optimization problem 
\begin{equation}
v(x) =\sup_{a \in \mathcal{A}}E_{\mathbb{P}}\left[ \left.
-e^{-\frac{1}{\delta }x_{T}}\right \vert x_{0}=x\right] ,
\label{v-classical}
\end{equation}%
with $\delta >0$ and $(x_{t})_{t\in \left[ 0,T\right] }$
solving 
\begin{equation}
\,dx_{t}=a_{t}\left( \mu _{t}\,dt+\sigma _{t}\,dW_{t}\right) ,  \quad x_{0}=x\in \mathbb{R}, \quad a \in \mathcal{A}. \label{x-s}
\end{equation}%
Then, the optimal
policy $\left( a_{t}^{\ast }\right) _{t\in \left[ 0,T\right] }$ and the
value function are given by 
\begin{equation}
 a_{t}^{\ast }=\delta \left( \frac{\lambda _{t}}{\sigma _{t}}+\frac{\rho }{%
1-\rho ^{2}}\frac{\xi _{t}}{\sigma _{t}}\right) ,  \label{a-x-s-optimal}
\end{equation}
and
\begin{equation}
 v(x) =-e^{-\frac{1}{\delta }x}M_{0}^{\frac{1}{1-\rho ^{2}}} = -e^{-\frac{1}{\delta}x} \left(E_{\QQ^{MM}}\left[e^{-\half(1-\rho^2)\int_0^T \lambda_s^2 \, ds}\right]\right)^{\frac{1}{1-\rho^2}},
 \label{v-optimal}
\end{equation}
with $\left( \xi
_{t}\right) _{t\in \left[ 0,T\right] }$ as in \eqref{def_xi}. 
\end{lem}

\begin{proof}
 We only present the key steps, showing that the
process $\left( u_{t}\right) _{t\in \left[ 0,T\right] },$ 
\begin{equation*}
u_{t}:=-e^{-\frac{1}{\delta }x_{t}}\left( E_{\mathbb{Q}^{MM}}\left[ \left.
e^{-\frac{1}{2}(1-\rho ^{2})\int_{t}^{T}\lambda _{s}^{2}\,ds}\right \vert 
\mathcal{G}_{t}\right] \right) ^{\frac{1}{1-\rho ^{2}}},
\end{equation*}%
with $u_{0}=v(x) ,$ $x\in \mathbb{R},$ is a supermartingale for $%
x_{t}$ solving \eqref{x-s} for arbitrary $\alpha \in \mathcal{A}$ and
becomes a martingale for $\alpha ^{\ast }$ as in \eqref{a-x-s-optimal}. To
this end, we write 
\begin{equation*}
u_{t}=-e^{-\frac{x_{t}}{\delta }}M_{t}^{\frac{1}{1-\rho ^{2}}}e^{N_{t}}\quad 
\text{with \  \ }N_{t}=\frac{1}{2}\int_{0}^{t}\lambda _{u}^{2}\,du,
\end{equation*}%
and observe that 
\begin{align*}
\,du_{t}& =-\frac{u_{t}}{\delta }\,dx_{t}+\frac{1}{2\delta^2}u_t%
\,d\langle x\rangle _{t}+u_{t}\,dN_{t}+\frac{1}{1-\rho ^{2}}\frac{u_{t}}{%
M_{t}}\,dM_{t}\, \\
& \qquad +\frac{1}{2(1-\rho ^{2})}\frac{\rho ^{2}}{1-\rho ^{2}}\frac{u_{t}}{%
M_{t}^{2}}\,d\langle M\rangle _{t}-\frac{1}{\delta (1-\rho ^{2})}\frac{u_{t}%
}{M_{t}}\,d\langle x,M\rangle _{t} \\
& =u_{t}\left( -\frac{1}{\delta }a_{t}\mu _{t}+\frac{1}{2}\frac{1}{\delta
^{2}}a_{t}^{2}\sigma _{t}^{2}+\frac{1}{2}\lambda _{t}^{2}+\frac{\rho }{%
1-\rho ^{2}}\xi _{t}\lambda _{t}+\frac{\rho ^{2}}{2(1-\rho ^{2})^{2}}\xi
_{t}^{2}-\frac{\rho }{\delta (1-\rho ^{2})}a_{t}\sigma _{t}\xi _{t}\right) dt
\\
& \qquad +u_{t}\left( -\frac{1}{\delta }a_{t}\sigma _{t}\,dW_{t}+\frac{1}{%
1-\rho ^{2}}\xi _{t}\,dW_{t}^{Y}\right) \\
& =\frac{1}{2}u_{t}\left( -\frac{1}{\delta }\sigma _{t}a_{t}+\lambda _{t}+%
\frac{\rho }{1-\rho ^{2}}\xi _{t}\right) ^{2}dt+u_{t}\left( -\frac{1}{%
\delta }a_{t}\sigma _{t}\,dW_{t}+\frac{1}{1-\rho ^{2}}\xi
_{t}\,dW_{t}^{Y}\right).
\end{align*}%
Because $u_{t}<0,$ the drift remains non-positive and vanishes for $t\in %
\left[ 0,T\right] $ if and only if the policy 
\begin{equation*}
a_{t}^{\ast }=\delta \left( \frac{\lambda _{t}}{\sigma _{t}}+\frac{\rho }{%
1-\rho ^{2}}\frac{\xi _{t}}{\sigma _{t}}\right)
\end{equation*}%
is being used. Furthermore, $a^{\ast }\in \mathcal{A},$ as it follows from
the boundedness assumption on $\sigma$, inequality \eqref{lamda-bounds} and that $\xi
\in \mathcal{L}^{2}\left( \mathbb{P}\right) .$ The rest of the proof follows
easily. 
\end{proof}

Next, we present the first main result herein that yields the existence of a
(wealth-independent) stochastic Nash equilibrium.

\begin{prop}
For $\delta _{i}>0$ and $c_{i}\in (-\infty
,1]$, introduce the quantities 
\begin{equation}
\varphi _{N}:=\frac{1}{N}\sum_{i=1}^{N}\delta _{i}\text{ \  \  \  \ and \  \  \ }%
\psi _{N}:=\frac{1}{N}\sum_{i=1}^{N}c_{i},  \label{def_phipsi}
\end{equation}%
and 
\begin{equation}
\bar{\delta}_{i}:=\delta _{i}+\frac{\varphi _{N}}{1-\psi _{N}}c_{i}.
\label{risk-aversion-N}
\end{equation}%
The following assertions hold:

\begin{enumerate}
\item If $\psi _{N}<1$, there exists a wealth-independent Nash equilibrium, $%
\left( \pi _{t}^{\ast }\right) _{t\in \left[ 0,T\right] }=\left( \pi
_{t}^{1,\ast }, \ldots ,\pi _{t}^{i,\ast }, \ldots ,\pi _{t}^{N,\ast }\right) _{t\in %
\left[ 0,T\right] }$, where $\pi _{t}^{i,\ast }$, $i\in \mathcal{I}$, is
given by the $\mathcal{G}_{t}$-adapted process 
\begin{equation}
\pi _{t}^{i,\ast }=\bar{\delta}_{i}\left( \frac{\lambda _{t}}{\sigma _{t}}+%
\frac{\rho }{1-\rho ^{2}}\frac{\xi _{t}}{\sigma _{t}}\right) ,
\label{pi-optimal-1}
\end{equation}%
with $(\xi _{t}) _{t\in \left[ 0,T\right] }$ as in \eqref{def_xi}. The associated optimal wealth process $\left( X_{t}^{i,\ast }\right)
_{t\in \left[ 0,T\right] }$ is 
\begin{equation}
X_{t}^{i,\ast }=x_{i}+\bar{\delta}_{i}\int_{0}^{t}\left( \lambda _{u}+\frac{%
\rho }{1-\rho ^{2}}\xi _{u}\right) \left( \lambda _{u}\,du+\,dW_{u}\right)
\label{X-optimal-1}
\end{equation}%
and the game value for player $i$, $i\in \mathcal{I}$, is given by 
\begin{align}
V^{i}\left( x_{1},x_{2},\ldots ,x_{N}\right) &=-\exp \left( -\frac{1}{{\delta 
}_{i}}\left( x_{i}-c_{i}\bar{x}\right) \right) M_{0}^{\frac{1}{1-\rho ^{2}}} \nonumber\\
&= -\exp \left( -\frac{1}{{\delta 
}_{i}}\left( x_{i}-c_{i}\bar{x}\right) \right)\left(E_{\QQ^{MM}}\left[e^{-\half(1-\rho^2)\int_0^T \lambda_s^2 \,ds} \right]\right)^{\frac{1}{1-\rho^2}},
\label{V-opimal-1}
\end{align}%
with $\bar{x}=\frac{1}{N}\sum_{i=1}^{N}x_{i}$.

\item If $\psi _{N}=1$, then it must be that $c_{i}\equiv 1$, for all $i\in 
\mathcal{I}$, and there is no such wealth-independent Nash equilibrium.
\end{enumerate}
\end{prop}

\begin{proof}
 We first solve the individual optimization problem %
\eqref{i-utility} for player $i\in \mathcal{I}$, taking the (arbitrary)
strategies $(\pi ^{1},\ldots ,\pi ^{i-1},\pi ^{i+1},\ldots ,\pi ^{N})$ of
all other players as given. This problem can be
alternatively written as 
\begin{equation}
v^{i}\left( \tilde{x}_{i}\right) =\sup_{\widetilde{\pi }^{i}\in \mathcal{A}}%
\mathbb{E}_{\mathbb{P}}\left[ \left. -\exp \left( -\frac{1}{\delta _{i}}%
\tilde{x}_{T}^{i}\right) \right \vert \tilde{x}_{0}^{i}=\tilde{x}_{i}\right]
,  \label{i-utility-aux}
\end{equation}%
where $\tilde{x}_{t}^{i}:=X_{t}^{i}-\frac{c_{i}}{N}\sum_{j=1}^{N}X_{t}^{j}$, 
$t\in \left[ 0,T\right] ,$ solves 
\begin{equation*}
\,d\tilde{x}_{t}^{i}=\widetilde{\pi }_{t}^{i}\left( \mu _{t}\,dt+\sigma
_{t}\,dW_{t}\right) \text{ \  \  \ and \  \ }\tilde{x}_{0}^{i}=\tilde{x}%
_{i}:=x_{i}-c_{i}\bar{x}.
\end{equation*}%
From Lemma~\ref{lem:nocompetition}, we deduce that its optimal policy is given by 
\begin{equation*}
\widetilde{\pi }_{t}^{i,\ast }=\delta _{i}\left( \frac{\lambda _{t}}{\sigma
_{t}}+\frac{\rho }{1-\rho ^{2}}\frac{\xi _{t}}{\sigma _{t}}\right) ,
\end{equation*}%
and thus the optimal policy of \eqref{i-utility} can be written as 
\begin{equation}
\pi _{t}^{i,\ast }=\delta _{i}\left( \frac{\lambda _{t}}{\sigma _{t}}+\frac{%
\rho }{1-\rho ^{2}}\frac{\xi _{t}}{\sigma _{t}}\right) +\frac{c_{i}}{N}%
\left( \sum_{j\neq i}\pi _{t}^{j}+\pi _{t}^{i,\ast }\right) .  \label{pi-*}
\end{equation}%
Symmetrically, all players $j\in \mathcal{I}$ follow an analogous to %
\eqref{pi-*} strategy. Averaging over $j\in \mathcal{I}$ yields 
\begin{equation*}
\frac{1}{N}\sum_{i=1}^{N}\pi _{t}^{i,\ast }=\psi _{N}\frac{1}{N}%
\sum_{i=1}^{N}\pi _{t}^{i,\ast }+\varphi _{N}\left( \frac{\lambda _{t}}{%
\sigma _{t}}+\frac{\rho }{1-\rho ^{2}}\frac{\xi _{t}}{\sigma _{t}}\right) ,
\end{equation*}%
with $\psi _{N}$ and $\varphi _{N}$ as in \eqref{def_phipsi}. 
If $\psi _{N}<1$, the above equation gives 
\begin{equation*}
\frac{1}{N}\sum_{i=1}^{N}\pi _{t}^{i,\ast }=\frac{\varphi _{N}}{1-\psi _{N}}%
\left( \frac{\lambda _{t}}{\sigma _{t}}+\frac{\rho }{1-\rho ^{2}}\frac{\xi
_{t}}{\sigma _{t}}\right) ,
\end{equation*}%
and we obtain \eqref{pi-optimal-1}. The rest of the proof follows easily. 
\end{proof}

We have stated the above result assuming that we start at $t=0$. This is
without loss of generality, as all arguments may be modified accordingly.
For completeness, we present in the
sequel the time-dependent case, in the context of a Markovian market. 

\begin{rem}
As discussed in \cite[Remark~2.5]{LaZa:17}, problem \eqref{i-utility} may be alternatively and equivalently represented
as 
\begin{equation*}
V^{i}(x_{1}, \ldots, x_{N}) =\sup_{\pi ^{i}\in \mathcal{A}}E_{\mathbb{%
P}}\left[ \left. -\exp \left( -\frac{1}{\delta _{i}^{\prime }}\left(
X_{T}^{i}-c_{i}^{\prime }C_{T}^{-i}\right) \right) \right \vert
X_{0}^{1}=x_{1}, \ldots, X_{0}^{i}=x_{i}, \ldots, X_{0}^{N}=x_{N}\right] ,
\end{equation*}%
with $C_{T}^{-i}:=\frac{1}{N-1}\sum_{j=1,j\neq i}^{N}X_{T}^{j}$, and $\delta
_{i}=\frac{\delta _{i}^{\prime }}{1+\frac{1}{N-1}c_{i}^{\prime }}$ and $%
c_{i}=\frac{c_{i}^{\prime }}{\frac{N-1}{N}+\frac{c_{i}^{\prime }}{N}}$.
\end{rem}

\begin{rem}
Instead of working with the minimal martingale measure in
the incomplete It\^{o}-diffusion market herein, one may employ the minimal
entropy measure, $\mathbb{Q}^{ME}$, given by 
\begin{equation}\label{def_ME}
\frac{d\mathbb{Q}^{ME}}{d\mathbb{P}}=\exp \left( -\frac{1}{2}%
\int_{0}^{T}\left( \lambda _{s}^{2}+\chi _{s}^{2}\right)
\,ds-\int_{0}^{T}\lambda _{s}\,dW_{s}-\int_{0}^{T}\chi _{s}\,dW_{s}^{\perp
}\  \right) ,
\end{equation}%
where $\chi _{t}= -Z_{t}^{\perp }$ and $\left( y_{t},Z_{t},Z_{t}^{\perp
}\right) _{t\in \left[ 0,T\right] }$ solves the backward stochastic differential equation (BSDE)
\begin{equation}\label{ME_BSDE}
- dy_{t}=\left( -\frac{1}{2}\lambda _{t}^{2}+\frac{1}{2}(Z_{t}^{\perp
})^{2}-\lambda _{t}Z_{t}\right) dt-\left( Z_{t}\,dW_{t}+Z_{t}^{\perp
}\,dW_{t}^{\perp }\right) \text{ \ and \  \ }y_{T}=0.
\end{equation}%
The measures $\mathbb{Q}^{ME}$ and $\mathbb{Q}^{MM}$ are related through the
relative entropy $\mathcal{H}$ in that $-\mathcal{H}(\mathbb{Q}^{ME}|\mathbb{%
P})=\frac{1}{1-\rho ^{2}}\ln M_{0}$ (cf. \cite{RoEl:00}). We choose to work
with $\mathbb{Q}^{MM}$ for ease of the presentation.
\end{rem}

From Lemma~\ref{lem:nocompetition}, we see that the
Nash equilibrium process, 
\begin{equation*}
\pi _{t}^{i,\ast }=\bar{\delta}_{i}\left( \frac{\lambda _{t}}{\sigma _{t}}+%
\frac{\rho }{1-\rho ^{2}}\frac{\xi _{t}}{\sigma _{t}}\right),
\end{equation*}%
resembles the optimal policy of an individual player of the classical optimal
investment problem with exponential utility and \textit{modified} risk tolerance, $\bar{\delta}_{i}.$
The latter deviates from $\delta _{i}$ by 
\begin{equation*}
\bar{\delta}_{i}-\delta _{i}=\frac{\varphi _{N}}{1-\psi _{N}}c_{i}.
\end{equation*}

In the competitive case, $c_{i}>0$, $\bar{\delta}_{i}>\delta _{i}$ and
their difference increases with $c_{i}$, $\varphi _{N}$ and $\psi _{N}$. At
times $t$ such that $\frac{\lambda _{t}}{\sigma _{t}}+\frac{\rho }{1-\rho
^{2}}\frac{\xi _{t}}{\sigma _{t}}>0$ (resp. $\frac{\lambda _{t}}{\sigma _{t}}%
+\frac{\rho }{1-\rho ^{2}}\frac{\xi _{t}}{\sigma _{t}}$ $ < 0)$, the
competition concerns make the player  invest more (resp. less) in the risky
asset than without such concerns.

In the homophilous case, $c_{i}<0$, we have that $\bar{\delta}_{i}<\delta _{i}$. Furthermore, direct computations show
that their difference decreases with $\delta _{i}$ and each $c_{j}$, $j\neq
i$, while it increases with $c_{i}$. In other words, 
\begin{equation*}
\partial _{\delta _{j}}\left( \bar{\delta}_{i}-\delta _{i}\right) <0,\text{ }%
\forall j\in \mathcal{I},\text{ \  \ }\partial _{c_{j}}\left( \bar{\delta}%
_{i}-\delta _{i}\right) <0,\text{ }\forall j\in \mathcal{I\smallsetminus }%
\left \{ i\right \} \text{, \  \ and \ \ }\partial _{c_{i}}\left( \bar{\delta}%
_{i}-\delta _{i}\right) >0.
\end{equation*}%
At times $t$ such that $\frac{\lambda _{t}}{\sigma _{t}}+\frac{\rho }{1-\rho
^{2}}\frac{\xi _{t}}{\sigma _{t}}>0,$ the player would invest less in the
risky asset, compared to without homophilous interaction. This investment decreases if other players become more risk tolerant (their $\delta
_{j}^{\prime }$ $s$ increase) or less homophilous (their $c_{j}^{\prime }$ $%
s $ increase)  or if the specific player $i$ becomes more homophilous ($c_{i}$
decreases). The case $\frac{\lambda _{t}}{\sigma _{t}}+\frac{\rho }{1-\rho
^{2}}\frac{\xi _{t}}{\sigma _{t}}<0$ follows similarly. The comparison
between the competitive and the homophilous case is described in Figure 1.

\begin{figure}[h]
	\centering
	\includegraphics[width = 0.9\textwidth, trim = {14em 4em 14em 6em}, clip, keepaspectratio=True]{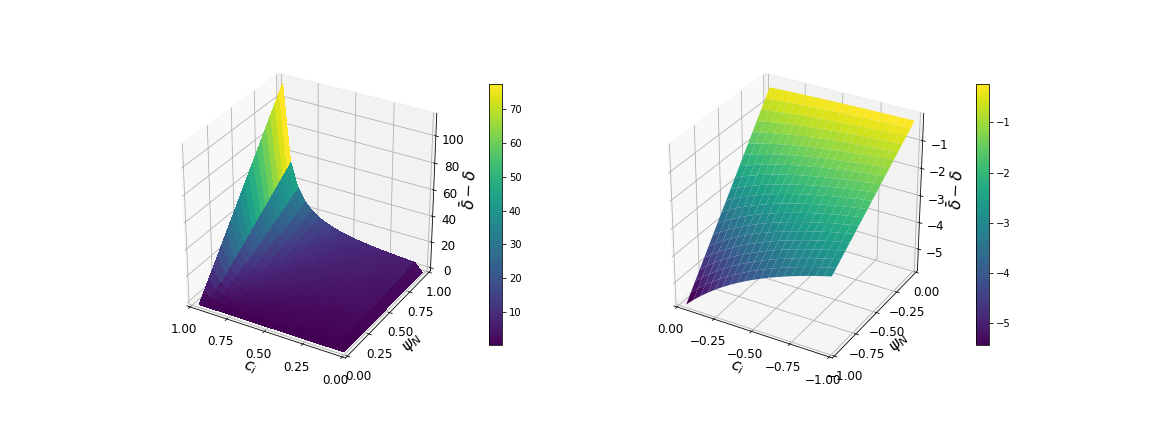}
	\caption{The plot of $\bar \delta_i - \delta_i$ versus $c_i$ and $\psi_N$, with $N = 25$ and $\varphi_N = 6$.}\label{fig:diff}
\end{figure}

\subsubsection{The Markovian case}

\label{sec_markov}

We consider a single stochastic factor model in which the stock price process $\left( S_{t}\right) _{t\in \left[ 0,T\right] }$ solves 
\begin{align}
\,dS_{t}& =\mu (t,Y_{t})S_{t}\,dt+\sigma (t,Y_{t})S_{t}\,dW_{t},
\label{def_St} \\
\,dY_{t}& =b(t,Y_{t})\,dt+a(t,Y_{t})\,dW_{t}^{Y},  \label{def_Yt}
\end{align}%
with $S_{0}=S>0$ and $Y_{0}=y\in \mathbb{R}$. The market coefficients $\mu, \sigma, a$ and  $b$ satisfy appropriate conditions for these
equations to have a unique strong solution. Further conditions, added next, are needed for the validity of the Feynman-Kac formula in Proposition~\ref{prop:markovian-incomplete}.

\begin{assump}\label{assump:incomplete}
The coefficients $\mu, \sigma, a$ and $b$ are bounded
functions, and $a, b$ have bounded, uniformly in $t$, $y$-derivatives. It is further assumed
that the Sharpe ratio function $\lambda (t,y):=\frac{\mu \left( t,y\right) 
}{\sigma \left( t,y\right) }$ is bounded and with bounded, uniformly in $t$,
$y$-derivatives of any order.
\end{assump}

For $t\in \left[ 0,T\right] ,$ we consider the optimization
problem 
\begin{multline}
V^{i}(t, x_{1}, \ldots, x_{i}, \ldots, x_{N}, y) \\=\sup_{\pi ^{i}\in \mathcal{A%
}}E_{\mathbb{P}}\left[ \left. -\exp \left( -\frac{1}{\delta _{i}}\left(
X_{T}^{i}-c_{i}C_{T}\right) \right) \right \vert
X_{t}^{1}=x_{1}, \ldots, X_{t}^{i}=x_{i}, \ldots, X_{t}^{N}=x_{N}, Y_{t}=y\right] ,
\label{V-markovian}
\end{multline}%
with $( X^{i}_s)_{s\in \left[ t,T\right] }$ solving $%
dX_{s}^{i}=\mu (s,Y_{s})\pi _{s}^{i}\,ds+\sigma (s,Y_{s})\pi
_{s}^{i}\,dW_{s} $ and $\pi ^{i}\in \mathcal{A},$ and $C_{T}$ as in \eqref{C-process}. We also consider the process $(\zeta _{t})_{t\in \left[ 0,T\right] }$
with $\zeta _{t}:=\zeta (t,Y_{t})$, where $\zeta :\left[ 0,T\right] \times 
\mathbb{R\rightarrow R}^{+}$ is defined as
\begin{equation*}
\zeta (t,y)=E_{\mathbb{Q}^{MM}}\left[ \left. e^{-\frac{1}{2}(1-\rho
^{2})\int_{t}^{T}\lambda ^{2}(s,Y_{s})\,ds}\right \vert Y_{t}=y\right] .
\end{equation*}%
Under $\mathbb{Q}^{MM}$, the stochastic factor process $\left( Y_{t}\right)
_{t\in \left[ 0,T\right] }$ satisfies 
\begin{equation*}
\,dY_{t}=( b(t,Y_{t})-\rho \lambda (t,Y_{t})a(t,Y_{t}))
\,dt+a(t,Y_{t})\,d\widetilde{W}_{t}^{Y}.
\end{equation*}%
Thus, using the conditions on the market coefficients and the Feynman-Kac
formula, we deduce that $\zeta (t,y)$ solves 
\begin{equation}
\zeta _{t}+\frac{1}{2}a^{2}(t,y)\zeta _{yy}+(b(t,y)-\rho \lambda
(t,y)a(t,y))\zeta _{y}=\frac{1}{2}(1-\rho ^{2})\lambda ^{2}(t,y)\zeta ,
\label{eq_zeta}
\end{equation}%
with $\zeta (T,y)=1.$ In turn, the function $f(t,y):=\frac{1}{1-\rho ^{2}}%
\ln \zeta (t,y)$ satisfies
\begin{equation}
f_{t}+\frac{1}{2}a^{2}(t,y)f_{yy}+(b(t,y)-\rho \lambda (t,y)a(t,y))f_{y}+%
\frac{1}{2}(1-\rho ^{2})a^{2}(t,y)f_{y}{}^{2}=\frac{1}{2}\lambda
^{2}(t,y),\quad f(T,y)=0.  \label{eq_f}
\end{equation}%
In the absence of competitive/homophilous interaction, this
problem has been examined by various authors (see, for example, \cite{SiZa:05}).

\begin{prop}\label{prop:markovian-incomplete}
Under Assumption~\ref{assump:incomplete}, the following
assertions hold for $t\in \left[ 0,T\right] .$
\begin{enumerate}
	\item  If $\psi _{N}<1$, there exists a wealth-independent Nash equilibrium $\left( \pi ^{\ast }_s\right)_{_{s\in \left[ t,T\right] }}=\left( \pi ^{1,\ast}_s, \ldots, \pi ^{i,\ast }_s, \ldots, \pi ^{N,\ast }_s\right) _{_{s\in \left[ t,T\right] }},$
where $\pi _{s}^{i,\ast }$,  $i\in \mathcal{I}$, is given by the process 
\begin{equation}
\pi _{s}^{i,\ast }=\pi ^{i,\ast }(s,Y_{s}),  \label{pi-Markovian}
\end{equation}%
with $(Y_t)_{t \in [0,T]}$ solving \eqref{def_Yt} and $\pi ^{i,\ast }:\left[ 0,T\right] \times 
\mathbb{R\rightarrow R}$ defined as 
\begin{equation}
\pi ^{i,\ast }(t,y):=\bar{\delta}_{i}\left( \frac{\lambda (t,y)}{\sigma (t,y)%
}+\rho \frac{a(t,y)}{\sigma (t,y)}f_{y}(t,y)\right) ,
\label{def_pi_markovian}
\end{equation}%
with $\bar{\delta}_{i}$ as in \eqref{risk-aversion-N} and $f(t,y)$ solving %
\eqref{eq_f}. The game value of player $i$, $i\in \mathcal{I}$, is given by 
\begin{align*}
V^{i}(t, x_{1}, \ldots, x_{N}, y) &=-\exp \left( -\frac{1}{\delta _{i}}%
\left( x_{i}-\frac{c_{i}}{N}\Sigma _{i=1}^{N}x_{i}\right) \right) \zeta
(t,y)^{\frac{1}{1-\rho ^{2}}} \\
&=-\exp \left(  -\frac{1}{\delta _{i}}\left( x_{i}-\frac{c_{i}}{N}%
\Sigma _{i=1}^{N}x_{i}\right) +f\left( t,y\right) \right) .
\end{align*}

\item If $\psi _{N}=1$, there exists no such Nash equilibrium.
\end{enumerate}
\end{prop}

\begin{proof}
To ease the notation, we establish the results when $t=0$ in \eqref{V-markovian}. To this end, we first identify the process $\xi $ in \eqref{def_xi}. For this, we rewrite the martingale in \eqref{martingale} as 
\begin{equation*}
M_{t}=\zeta (t,Y_{t})e^{-\frac{1}{2}(1-\rho ^{2})\int_{0}^{t}\lambda
^{2}(s,Y_{s})\,ds},
\end{equation*}
and observe that 
\begin{align}
\,dM_{t}& =\left( \zeta _{t}(t,Y_{t})+(b(t,Y_{t})-\rho a(t,Y_{t})\lambda
(t,Y_{t}))\zeta _{y}(t,Y_{t})+\frac{1}{2}a^{2}(t,Y_{t})\zeta
_{yy}(t,Y_{t})\right) \frac{M_{t}}{\zeta (t,Y_{t})}\,dt  \notag \\
& \qquad -\frac{1}{2}(1-\rho ^{2})\lambda ^{2}(t,Y_{t})M_{t}\,dt+a(t,Y_{t})%
\frac{\zeta _{y}(t,Y_{t})}{\zeta (t,Y_{t})}M_{t}\left( \rho \,d\widetilde{W}%
_{t}+\sqrt{1-\rho ^{2}}\,dW_{t}^{\perp }\right) \\
& =a(t,Y_{t})\frac{\zeta _{y}(t,Y_{t})}{\zeta (t,Y_{t})}M_{t}\left( \rho \,d%
\widetilde{W}_{t}+\sqrt{1-\rho ^{2}}\,dW_{t}^{\perp }\right) ,
\end{align}%
where we used that $\zeta(t,y)$ satisfies \eqref{eq_zeta}. Therefore, $\xi
_{t}=a(t,Y_{t})\frac{\zeta _{y}(t,Y_{t}) }{\zeta(
t,Y_{t}) }$. In turn, using that $\zeta (t,y)^{1/(1-\rho
^{2})}=e^{f(t,y)},$ we obtain that 
\begin{equation*}
f_{y}(t,Y_{t})=\frac{1}{1-\rho ^{2}}\frac{\zeta _{y}(t,y) }{%
\zeta (t,y) }\text{ \  \ and \  \ }\xi _{t}=(1-\rho
^{2})a(t,Y_{t})f_{y}(t,Y_{t}),
\end{equation*}%
and we easily conclude by replacing $\xi _{t}$ by $(1-\rho
^{2})a(t,Y_{t})f_{y}(t,Y_{t})$ in \eqref{pi-optimal-1}.

It remains to show that the candidate investment process in \eqref{pi-Markovian} is admissible. Under Assumption~\ref{assump:incomplete} we deduce that $f_{y}( t,y)$ is a bounded function, since $\zeta (t,y)$
is bounded away from zero and $\zeta _{y}(t,y) $ is bounded. We easily conclude.
\end{proof}

\begin{rem}

 In the Markovian model \eqref{def_St}--\eqref{def_Yt}, the density of the minimal entropy
measure $\QQ^{ME}$ is fully specified. Indeed, the BSDE \eqref{ME_BSDE} admits the
solution 
\begin{equation*}
y_{t}=f(t,Y_{t}),\quad Z_{t}=\rho a(t,Y_{t})f_{y}(t,Y_{t})\text{ \  \ and
\  \ }Z_{t}^{\perp }=\sqrt{1-\rho ^{2}}a(t,Y_{t})f_{y}(t,Y_{t}),
\end{equation*}%
and, thus, the density of $\mathbb{Q}^{ME}$ is
given by \eqref{def_ME} with $\chi _{t}\equiv \chi (t,Y_{t})=-\sqrt{1-\rho
^{2}}a(t,Y_{t})f_{y}(t,Y_{t})$. 
\end{rem}

\subsubsection{A fully solvable example}

Consider the family of models with autonomous dynamics 
\begin{equation*}
\mu (t,y)=\mu y^{\frac{1}{2\ell }+\frac{1}{2}}\text{, \ }\sigma (t,y)=y^{%
\frac{1}{2\ell }},\text{ \ }b(t,y)=m-y,\quad a(t,y)=\beta \sqrt{y},\text{\ }
\end{equation*}%
with $\mu >0,$ $\beta >0,$ $\ell\neq 0$ and $m>\frac{1}{2}\beta ^{2}$. Notable
cases are $\ell =1,$ which corresponds to the Heston stochastic volatility
model, and $\ell =-1$ that is studied in \cite{ChVi:05}.

Equation \eqref{eq_f} depends only on $b(t,y)$, $a(t,y)$ and the Sharpe
ratio $\lambda (t,y)=$ $\mu \sqrt{y}$, and thus its solution $f(t,y)$ is
independent of the parameter $\ell $. Using the ansatz $f(t,y)=p(t)y+q(t)$ with $p(T)=q(T)=0$, we deduce from 
\eqref{eq_f} that $p(t)$ and $q(t)$ satisfy 
\begin{align}
& \dot{p}(t)-\frac{1}{2}(\mu +\rho \beta p(t))^{2}-p(t)+\frac{1}{2}\beta
^{2}p^{2}(t)=0,  \notag \\
& \dot{q}(t)+mp(t)=0.
\end{align}%
In turn, 
\begin{equation*}
p(t)=\frac{1+\rho \mu \beta -\sqrt{\Delta }}{(1-\rho ^{2})\beta ^{2}}\frac{%
1-e^{-\sqrt{\Delta }(T-t)}}{1-\frac{1+\rho \mu \beta -\sqrt{\Delta }}{1+\rho
\mu \beta +\sqrt{\Delta }}e^{-\sqrt{\Delta }(T-t)}},\quad \Delta =1+\beta
^{2}\mu ^{2}+2\rho \mu \beta >0,
\end{equation*}%
and $q(t)=m\int_{t}^{T}p(s)ds$.

From \eqref{def_pi_markovian}, we obtain that the Nash equilibrium strategy 
$\left( \pi _{s}^{i,\ast }\right) _{s\in \left[ t,T\right] },$ $t\in \left[
0,T\right] ,$ for player $i$ is given by the process 
\begin{equation*}
\pi _{s}^{i,\ast }=\bar{\delta}_{i}(\mu +\rho \beta p(s))Y_{s}^{\frac{1}{2}%
(1-\frac{1}{\ell })}.
\end{equation*}%
If $\ell =1$, the policy becomes deterministic, $\pi _{s}^{i,\ast }=\bar{\delta}_{i}(\mu +\rho \beta p(s)),$ and the equilibrium wealth process
solves 
\begin{equation*}
\,dX_{s}^{i,\ast }=\overline{\delta }_{i}(\mu +\rho \beta p(s))(\mu
Y_{s}\,ds+\sqrt{Y_{s}}\,dW_{s}).
\end{equation*}

\subsection{The common-noise MFG}\label{sec_MFG-incomplete}

We analyze the limit as $N\uparrow \infty $ of the $N$-player game studied in
Section~\ref{sec_Nagent}. We first give an intuitive and informal argument that
leads to a candidate optimal strategy in the mean-field setting, and then
propose a rigorous formulation for the MFG. The analysis follows closely the
arguments developed in \cite{LaZa:17}. 

For the $N$-player game, we denote by $\eta _{i}=(x_{i},\delta _{i},c_{i})$
the \emph{type vector} for player $i$, where $x_{i}$ is her initial wealth,
and $\eta _{i}$ and $c_{i}$ are her risk tolerance
coefficient and interaction parameter, respectively. Such type vectors induce an empirical
measure $m_{N}$, called the \emph{type distribution}, 
\begin{equation*}
m_{N}(A)=\frac{1}{N}\sum_{i=1}^{N}\mathbf{1}_{\eta _{i}}(A),\text{ for Borel
sets }A\subset \mathcal{Z},
\end{equation*}%
which is a probability measure on the space $\mathcal{Z}:=\mathbb{R}\times
(0,\infty )\times (-\infty, 1]$.

We recall ({\it cf.} \eqref{pi-optimal-1}) that the equilibrium strategies $(\pi _{t}^{i,\ast }) _{t\in \left[ 0,T\right] }$, $i\in \mathcal{I}$, 
are given as the product of the common (type-independent) process $\frac{\lambda_t}{\sigma_t}
+\frac{\rho }{1-\rho ^{2}}\frac{\xi _{t}}{\sigma_t}$ and the modified risk tolerance
parameter $\bar{\delta}_{i}=\delta _{i}+\frac{\varphi _{N}}{1-\psi _{N}}%
c_{i} $. Therefore, it is \textit{only} the coefficient $\bar{\delta}_{i}$
that depends on the empirical distribution $m_{N}$ through $\psi _{N}$ and $%
\varphi _{N}$, as both these quantities can be obtained by averaging
appropriate functions over $m_{N}$. Therefore, if we assume that $m_{N}$
converges weakly to some limiting probability measure as $N\uparrow \infty $%
, we should intuitively expect that the corresponding equilibrium strategies
also converge. This is possible, for instance, by letting the type vector $%
\eta =(x,\delta ,c)$ be a random variable in the space $\mathcal{Z}$ with
limiting distribution $m$, and take $\eta _{i}$ as i.i.d. samples of $\eta $%
. The sample $\eta _{i}$ is drawn and assigned to player $i$ at initial time $%
t=0.$ We would then expect $( \pi ^{i,\ast })_{t\in \left[ 0,T%
\right] }$ to converge to the process 
\begin{equation}
\lim_{N\uparrow \infty }\pi _{t}^{i,\ast }=\left( \delta _{i}+\frac{\bar \delta}{1-%
\bar{c}}c_i\right) \left( \frac{\lambda _{t}}{\sigma _{t}}+\frac{%
\rho }{1-\rho ^{2}}\frac{\xi _{t}}{\sigma _{t}}\right),  \label{def_pi_mfg}
\end{equation}%
where $\bar{c}$ and $\bar{\delta}$ represent the average interaction and
risk tolerance coefficients.

Next, we introduce the mean-field game in the incomplete It\^{o}-diffusion
market herein, and we show that \eqref{def_pi_mfg} indeed arises
as its equilibrium strategy. We model a single representative player, whose
type vector is a random variable with distribution $m$, and all players in
the continuum act in this common incomplete market.

\subsubsection{The It\^{o}-diffusion common-noise MFG}

\label{sec_MFG_general} To describe the heterogeneous population of players, we
introduce the type vector 
\begin{equation}
\eta =(x,\delta ,c)\in \mathcal{Z},  \label{type-vector-case1}
\end{equation}%
where $\delta >0$ and $c\in \left( -\infty ,1\right] $ represent the risk
tolerance coefficient and interaction parameter, and $x$ is the initial wealth.
This type vector is assumed to be independent of both $W$ and $W^{Y}$, which
drive the stock price process \eqref{stock-1}, and is assumed to have finite second moments.  

To formulate the mean-field portfolio game, we now let the filtered
probability space $(\Omega ,\mathcal{F},\mathbb{P}) $ support
$W,$ $W^{Y}$ as well as $\eta $. We assume that $\eta$ has second moments under $\PP$.
 We denote by $(\mathcal{F}_{t}^{MF})_{t\in \lbrack 0,T]}$ the smallest filtration satisfying the usual
assumptions for which $\eta$ is $\mathcal{F}_{0}^{MF}$-measurable and both 
$W,W^{Y}$ are adapted. As before, we denote by $(\mathcal{F}_{t})_{t\in \lbrack 0,T]}$ the natural filtration generated by $W$ and $%
W^{Y},$ and by $(\mathcal{G}_{t})_{t\in \lbrack 0,T]}$ the one
generated only by $W^{Y}.$

We also consider the wealth process $\left( X_{t}\right) _{t\in \left[ 0,T%
\right] }$ of the \emph{representative player} solving 
\begin{equation}
\,dX_{t}=\pi _{t}\left( \mu _{t}\,dt+\sigma _{t}\,dW_{t}\right) ,
\label{X-t-MFG}
\end{equation}%
with $X_{0}=x\in \mathbb{R}$ and $\pi $ $\in \mathcal{A}^{MF}$, where
\begin{equation*}
\mathcal{A}^{MF}=\left \{ \pi :\text{self-financing, }\mathcal{F}_{t}^{MF}%
\text{-progressively measurable and }E_{\mathbb{P}}\left[ \int_{0}^{T}\sigma_s^2\pi
_{s}^{2}\,ds\right] <\infty \right \} .
\end{equation*}

Similarly to the framework in \cite{LaZa:17}, there exist two independent
sources of randomness in the model: the first is due to the evolution of the
stock price process, described by the Brownian motions $W$ and $W^{Y}$. The
second is given by $\eta $, which models the type of the player, {\it i.e.}, the
triplet of initial wealth, risk tolerance, and interaction parameter in the
population continuum. The first source of noise is \textit{stochastic} and
common to each player in the continuum while the second is \textit{static}%
, being assigned at time zero and with the dynamic competition starting
right afterwards.

In analogy to the $N$-player setting, the representative player optimizes the expected terminal utility,
taking into account the performance of the average terminal wealth of the
population, denoted by $\overline{X}.$ As in \cite{LaZa:17}, we introduce the
following definition for the MFG considered herein.

\begin{defi}\label{def:MFG-incomplete}
For each $\pi \in \mathcal{A}^{MF}$, let $\overline{X}:=E_{\mathbb{P}}[X_{T}|\mathcal{F}_{T}]$ with $\left( X_{t}\right) _{t\in \left[ 0,T\right] }$ solving \eqref
	{X-t-MFG}, and consider the optimization problem
\begin{equation}
V(x) =\sup_{\pi \in \mathcal{A}^{MF}}E_{\mathbb{P}}\left[\left. -\exp
\left( -\frac{1}{\delta }\left( X_{T}-c\overline{X}\right) \right) \right \vert 
\mathcal{F}_{0}^{MF},X_{0}=x\right].  \label{V-MFG}
\end{equation}
A strategy $\mathit{\ }\pi ^{\ast }\in \mathcal{A}^{MF}$ is a mean-field equilibrium if $\pi ^{\ast }$ is the optimal strategy of the above problem when $\overline{X}^{\ast }:=E_{\mathbb{P}
}[X_{T}^{\ast }|\mathcal{F}_{T}]$ is used for $\overline X$, where $\left( X_{t}^{\ast }\right) _{t\in %
\left[ 0,T\right] }$ solves  \eqref{X-t-MFG} with $\pi ^{\ast }$ being used.
\end{defi}

Next, we state the main result.

\begin{prop}
If $E_{\mathbb{P}}[c] <1$, there exists a
unique wealth-independent MFG\ equilibrium $\left( \pi _{t}^{\ast }\right)
_{t\in \left[ 0,T\right] }$, given by the $\mathcal{F}_{0}^{MF}\vee \mathcal{%
G}_{t}$ process 
\begin{equation}
\pi _{t}^{\ast }=\left( \delta +\frac{E_{
		\mathbb{P}}[ \delta ] }{1-E_{\mathbb{P}}[c]}c\right) \left( \frac{\lambda _{t}}{\sigma
_{t}}+\frac{\rho }{1-\rho ^{2}}\frac{\xi _{t}}{\sigma _{t}}\right) ,
\label{pi-MFG}
\end{equation}%
with $\xi $ as in \eqref{def_xi}. The corresponding optimal wealth is given
by 
\begin{equation}\label{Xoptimal-MFG}
X_{t}^{\ast }=x+\left( \delta +\frac{E_{
		\mathbb{P}}[ \delta ] }{1-E_{\mathbb{P}}[c] }c \right)  \int_{0}^{t}\left(\lambda _{s} + \frac{\rho}{1-\rho^2}\xi_s\right)\left(
\lambda _{s}\,ds+\,dW_{s}\right),
\end{equation}%
and 
\begin{equation*}
V(x)=-\exp \left( -\frac{1}{\delta }\left( x-cm\right) \right) M_{0}^{\frac{1%
}{1-\rho ^{2}}} = -\exp \left( -\frac{1}{\delta }\left( x-cm\right) \right)\left(E_{\QQ^{MM}}\left[e^{-\half(1-\rho^2)\int_0^T \lambda_s^2 \,ds} \right]\right)^{\frac{1}{1-\rho^2}},
\end{equation*}
where $m = E_{\mathbb{P}}[x]$. 
If $E_{\mathbb{P}}[c] =1$, there is no such Nash equilibrium.
\end{prop}

\begin{proof}
We first observe that $\pi^{\ast }$ in \eqref{pi-MFG} is $%
\mathcal{F}_{t}^{MF}$-measurable since $\left( \frac{\lambda _{t}}{\sigma
_{t}}+\frac{\rho }{1-\rho ^{2}}\frac{\xi _{t}}{\sigma _{t}}\right) \in $ $%
\mathcal{G}_{t}$, and thus \\
$\left( \frac{\lambda _{t}}{\sigma _{t}}+\frac{%
	\rho }{1-\rho ^{2}}\frac{\xi _{t}}{\sigma _{t}}\right) \in \mathcal{F}_{t}$,
while the factor $\left( \delta +\frac{E_{
		\mathbb{P}}[ \delta ]}{1-E_{\mathbb{P}}[c] }c \right)  \in \mathcal{F}_{0}^{MF}$
(independent of $\mathcal{F}_{t}$). 
Furthermore, $\pi^{\ast }$ is also square-integrable under standing assumptions, and thus admissible. To show that it is also indeed an equilibrium policy, we shall first define $\overline X$ using  $\pi^\ast$, and then verify that the optimal strategy to the representative player's problem \eqref{V-MFG} coincides with $\pi_t^\ast$ when this specific $\overline X$ is used in \eqref{V-MFG}. To this end, we introduce the process $\overline X_t := E_\mathbb{P}[X_t^\ast \vert \mc{F}_t]$ with $(X_t^\ast)_{t \in [0,T]}$ as in  \eqref{Xoptimal-MFG}. Then,
\begin{align*}
\overline X_t & = E_\mathbb{P}\left[x+\left( \delta +\frac{E_{
		\mathbb{P}}[ \delta ]}{1-E_{\mathbb{P}}[c] }c\right)  \int_{0}^{t}\left(\lambda _{s} + \frac{\rho}{1-\rho^2}\xi_s\right)\left(
\lambda _{s}\,ds+\,dW_{s}\right)\Big\vert \mc{F}_t\right]\\
& = m +\left( E_\mathbb{P}[\delta] +\frac{E_\mathbb{P}[\delta]}{1-E_{\mathbb{P}}[ c] }E_{%
	\mathbb{P}}[ c ] \right) \int_{0}^{t}\left(\lambda _{s} + \frac{\rho}{1-\rho^2}\xi_s\right)\left(
\lambda _{s}\,ds+\,dW_{s}\right) \\
& = m +\left(\frac{E_{
		\mathbb{P}}[ \delta ]}{1-E_{\mathbb{P}}[ c] }\right) \int_{0}^{t}\left(\lambda _{s} + \frac{\rho}{1-\rho^2}\xi_s\right)\left(
\lambda _{s}\,ds + \,dW_{s}\right),
\end{align*}
where we have used that $\int_{0}^{t}\left(\lambda _{s} + \frac{\rho}{1-\rho^2}\xi_s\right)\left(
\lambda _{s}\, ds+ \,dW_{s}\right)$ is $\mc{G}_t$-measurable and thus $\mc{F}_t$-measurable, and that $\left( \delta +\frac{E_{
		\mathbb{P}}[ \delta ]}{1-E_{\mathbb{P}}[ c] }c\right)$ is independent of $\mc{F}_t$. 

Next, we introduce the auxiliary process  $\left( \tilde{x}_{t}\right)
_{t\in \left[ 0,T\right] }$,  $\, \tilde{x}_{t}:=\,X_{t}-c\overline{X}_{t},$ with $\left( X_{t}\right)
_{t\in \left[ 0,T\right] }$ as in \eqref{X-t-MFG}. Then, 
\begin{equation*}
d\tilde{x}_{t}=\widetilde{\pi }_{t}(\mu _{t}\,dt+\sigma _{t}\,dW_{t})\text{ \  \
and \  \ }\, \tilde{x}_{0} = \tilde x := x-cm,
\end{equation*}%
and $\widetilde{\pi }_{t}=\pi _{t}-c\left(\frac{E_{
		\mathbb{P}}[ \delta ]}{1-E_{\mathbb{P}}[ c] }\right) \left(\frac{\lambda _{t}}{\sigma_t} + \frac{\rho}{1-\rho^2}\frac{\xi_t}{\sigma_t}\right)$. In turn, we consider the
optimization problem 
\begin{equation*}
v(\tilde{x}) :=\sup_{\widetilde\pi \in \mathcal{A}^{MF}}E_{\mathbb{P}}\left[
\left. -\exp \left( -\frac{1}{\delta }\tilde{x}_{T}\right) \right \vert 
\mathcal{F}_{0}^{MF},\, \tilde{x}_{0}=\tilde{x}\right] .
\end{equation*}%
From Lemma~\ref{lem:nocompetition}, we deduce that the optimal strategy is given by 
\begin{equation*}
\widetilde{\pi }_{t}^{\ast }=\delta \left( \frac{\lambda _{t}}{\sigma _{t}}+%
\frac{\rho }{1-\rho ^{2}}\frac{\xi _{t}}{\sigma _{t}}\right),
\end{equation*}%
and, thus,
\begin{equation*}
\pi _{t}^{\ast }=\delta \left( \frac{\lambda _{t}}{\sigma _{t}}+\frac{\rho }{%
1-\rho ^{2}}\frac{\xi _{t}}{\sigma _{t}}\right) +c\left(\frac{E_{
	\mathbb{P}}[ \delta ]}{1-E_{\mathbb{P}}[c] }\right)\left( \frac{\lambda _{t}}{\sigma _{t}}+\frac{\rho }{%
1-\rho ^{2}}\frac{\xi _{t}}{\sigma _{t}}\right).
\end{equation*}%
The rest of the proof follows easily.
\end{proof}

If we view $\eta =(x,\delta ,c)$ in the $N$-player game in Section~\ref{sec_Nagent} as i.i.d. samples on the space $\mathcal{Z}$ with distribution $m$, then $\lim_{N\uparrow \infty }\psi_{N}=E_{\mathbb{P}}[c]$ and $\lim_{N\uparrow \infty }\varphi _{N}= E_{\mathbb{P}}[\delta] $ a.s.. We then obtain the convergence of the corresponding optimal processes, namely, for $t\in \left[ 0,T\right]$, 
\begin{equation*}
\lim_{N\uparrow \infty }\pi _{t}^{i,\ast }=\pi _{t}^{\ast },\text{ \  \ and \ } \lim_{N\uparrow \infty }X_{t}^{i,\ast }=X_{t}^{\ast }.
\end{equation*}

\subsubsection{The Markovian case}

In analogy to the $N$-player case, we have the following result.

\begin{prop}
Assume that the stock price process follows the single factor
model \eqref{def_St}--\eqref{def_Yt}. Then, if $E_{\mathbb{P}}[c] <1,$ there exists a unique wealth-independent Markovian mean-field
game equilibrium, given by the process $\left( \pi _{t}^{\ast }\right)
_{t\in \left[ 0,T\right] },$ 
\begin{equation*}
\pi _{t}^{\ast }=\pi ^{\ast }(\eta ,t,Y_{t})=\left( \delta +\frac{E_{
		\mathbb{P}}[ \delta ]}{1-E_{%
\mathbb{P}}[c] }c\right)
\left( \frac{\lambda (t,Y_{t})}{\sigma (t,Y_{t})}+\rho \frac{a(t,Y_{t})}{%
\sigma (t,Y_{t})}f_{y}(t,Y_{t})\right) ,
\end{equation*}%
with the $\mathcal{F}_{0}^{MF}$-measurable random function $\pi ^{\ast
}(\eta ,t,y):\mathcal{Z}\times \left[ 0,T\right] \times \mathbb{R},$ 
\begin{equation*}
\pi ^{\ast }(\eta ,t,y):=\left( \delta +\frac{E_{
		\mathbb{P}}[ \delta ]}{1-E_{\mathbb{P}}[c] }c \right) \left( \frac{\lambda
(t,y)}{\sigma (t,y)}+\rho \frac{a(t,y)}{\sigma (t,y)}f_{y}(t,y)\right) .
\end{equation*}%
If $E_\PP[c]=1$, there is no such mean-field game stochastic  equilibrium.
\end{prop}

\section{Complete It\^{o}-diffusion common market and CARA utilities with random
risk tolerance coefficients}\label{sec:complete}

In this section, we focus on the complete common market case, but we extend the model by allowing random individual risk tolerance coefficients. We start with a background result for the single-player problem, which is new and interesting in its own right. Building on it, we analyze both the $N$-player and the MFG. The analysis shows that the randomness of the individual risk tolerance gives rise to virtual ``personalized'' markets, in that the original common risk premium process now differs across players, depending on their risk tolerance. This brings substantial complexity as the tractability coming from the original common market is now lost.

\subsection{The It\^{o}-diffusion market and random risk tolerance coefficients}

We consider the complete analog of the It\^{o}-diffusion market studied in Section~\ref{sec:incomplete}. Specifically, we consider a market with a riskless bond
(taken to be the numeraire and offering zero interest rate) and a stock
whose price process $\left( S_{t}\right) _{t\in \left[ 0,T\right] }$ solves 
\begin{equation*}
dS_{t}=S_{t}\left( \mu _{t}\,dt+\sigma _{t}\,dW_{t}\right) ,
\end{equation*}%
with $S_{0}>0$, and $\left( W_t\right) _{t\in \left[ 0,T\right] }$ being a
Brownian motion in a probability space $(\Omega ,\mathcal{F},\mathbb{P})$.
The market coefficients $\left( \mu _{t}\right) _{t\in \left[ 0,T\right] }$
and $\left( \sigma _{t}\right) _{t\in \left[ 0,T\right] }$ are $\mathcal{F}%
_{t}$-adapted processes, where $(\mc{F}_t)_{t \in [0,T]}$ is the natural filtration generated by $W$, and with $0<c\leq \sigma _{t}\leq C$ and $\abs{\mu
_{t}}\leq C$, $t \in [0,T]$, for some (possibly deterministic) constants $c$ and $C$.

In this market, $N$ players, indexed by $i\in \mathcal{I}$, $\mathcal{I}%
=\left \{ 1, 2, \ldots, N\right \}$,  trade between the two accounts in $\left[ 0,T%
\right]$,  with individual wealths $\left( X_{t}^{i}\right) _{t\in \left[ 0,T%
\right] }$ solving 
\begin{equation}
\,dX_{t}^{i}=\pi _{t}^{i}\left( \mu _{t}\,dt+\sigma _{t}\,dW_{t}\right) ,
\label{X-i-complete}
\end{equation}%
and $X_{0}^{i}=x_{i}\in \mathbb{R}.$

Each of the players, say player $i$, has \textit{random risk tolerance}, $\delta
_{T}^{i}$, defined on $(\Omega ,\mathcal{F},\mathbb{P})$ with the
following properties:

\begin{assump}\label{assump:complete}
	For each $i\in \mathcal{I}$, the risk tolerance $\delta
	_{T}^{i}$ is an $\mathcal{F}_{T}$-measurable random variable with $\delta
	_{T}^{i}\geq \delta >0$ and $E_{\mathbb{P}}\left( \delta _{T}^{i}\right)
	^{2}<\infty$. 
\end{assump}

The objective of each player is to optimize 
\begin{multline}
V^{i}\left( x_{1}, \ldots, x_{i}, \ldots, x_{N}\right) \\
=\sup_{\mathcal{A}}E_{\mathbb{P}%
}\left[ \left. -\exp \left( -\frac{1}{\delta _{T}^{i}}\left( X_{T}^{i}-\frac{%
c_{i}}{N}\sum_{j=1}^{N}X_{T}^{j}\right) \right) \right \vert
X_{0}^{1}=x_{1}, \ldots, X_{0}^{i}=x_{i}, \ldots, X_{0}^{N}=x_{N}\right] ,
\label{V-i-risktolerance}
\end{multline}
with $c_{i}\in \left( -\infty ,1\right] $, $X^{j}$, $j\in \mathcal{I}$, solving
\eqref{X-i-complete}, and $\mathcal{A}$ defined similarly to \eqref%
{admissible-set-part1}.

As in Section~\ref{sec_Nagent}, we are interested in a Nash equilibrium
solution, which is defined as in Definition~\ref{def:Nash}. Before we solve
the underlying stochastic $N$-player game, we focus on the single-player case. This is a problem interesting in its own right and, to
our knowledge, has not been studied before in such markets. A similar
problem was considered in a single-period binomial model in \cite{musiela2002note} and in a special diffusion case in \cite{ringer2011three} in the context of
indifference pricing of bonds. For generality, we present below the
time-dependent case. 

\subsection{The single-player problem}\label{sec:single-complete}

We consider the optimization problem 
\begin{equation}
v_{t}( x) =\sup_{\pi \in \mathcal{A}}E_{\mathbb{P}}\left[\left. -e^{-%
\frac{1}{\delta _{T}}x_{T}}\right \vert \mathcal{F}_{t},  x_{t}=x\right],
\label{single-risk tolerance}
\end{equation}%
with $\delta _{T}\in \mathcal{F}_{T}$ satisfying Assumption~\ref{assump:complete} and 
$\left( x_{s}\right) _{s\in \left[ t,T\right] }$ solving \eqref{X-i-complete} with $x_{t}=x\in \mathbb{R}$.

We define $\left( Z_{t}\right) _{t\in \left[ 0,T\right] }$ by 
\begin{equation*}
Z_{t}=\exp \left( -\frac{1}{2}\int_{0}^{t}\lambda
_{s}^{2}\,ds-\int_{0}^{t}\lambda _{s}\,dW_{s}\right) ,
\end{equation*}%
and recall the associated (unique) risk neutral measure $\mathbb{Q}$,
defined on $\mathcal{F}_{T}$ and given by 
\begin{equation}
\frac{d\mathbb{Q}}{d\mathbb{P}}=Z_{T}.  \label{Q-measure}
\end{equation}%
We introduce the process $\left( \delta _{t}\right) _{t\in \left[ 0,T\right]
},$ 
\begin{equation}
\delta _{t}:=E_{\mathbb{Q}}[\delta _{T}|\mathcal{F}_{t}],
\label{delta-process}
\end{equation}%
which may be thought as the arbitrage-free price of the risk tolerance
``claim'' $\delta _{T}$. We also introduce the measure $\hat{\mathbb{Q}}$, defined on $\mathcal{F}%
_{T},$ with 
\begin{equation*}
\frac{d\hat{\mathbb{Q}}}{d\mathbb{P}} =\frac{\delta _{T}}{E_{\mathbb{Q}}[\delta _{T}]}Z_{T}.
\end{equation*}%
Direct calculations yield that under measure $\hat{\mathbb{Q}}$, the process 
$\left( \frac{S_{t}}{\delta _{t}}\right) _{t\in \left[ 0,T\right] }$ is an $%
\mathcal{F}_{t}$-martingale.

By the model assumptions and the martingale representation theorem, there
exists an $\mathcal{F}_{t}$-adapted process $\left( \xi _{t}\right) _{t\in %
\left[ 0,T\right] }$ with $\xi \in \mathcal{L}^{2}\left( \mathbb{P}\right) $
such that 
\begin{equation}
d\delta _{t}=\xi _{t}\delta _{t}\,dW_{t}^{\mathbb{Q}},  \label{xi-delta}
\end{equation}%
with $W_{t}^{\mathbb{Q}}=W_{t}+\int_{0}^{t}\lambda _{s}\,ds$. Next, we introduce
the process 
\begin{equation}
H_{t}:=E_{\mathbb{\widetilde{Q}}}\left[ \left. \frac{1}{2}\int_{t}^{T}\left(
\lambda _{s}-\xi _{s}\right) ^{2}\,ds\right \vert \mathcal{F}_{t}\right] ,
\label{def:H2}
\end{equation}%
where $\widetilde{\mathbb{Q}}$ is defined on $\mathcal{F}_{T}$ by 
\begin{equation}\label{def:Qtilde}
\frac{d\widetilde{\mathbb{Q}}}{d\mathbb{P}}=\exp \left( -\frac{1}{2}%
\int_{0}^{T}(\lambda _{s}-\xi _{s})^{2}\,ds-\int_{0}^{T}\left( \lambda
_{s}-\xi _{s}\right) \,dW_{s}\right) .
\end{equation}
Under $\widetilde{\mathbb{Q}}$, the process $\left( W_{t}^{\widetilde{%
\mathbb{Q}}}\right) _{t\in \left[ 0,T\right] }$ with 
\begin{equation}
W_{t}^{\widetilde{\mathbb{Q}}}:=W_{t}+\int_{0}^{t}\left( \lambda _{s}-\xi
_{s}\right) \,ds  \label{W-Q-telda}
\end{equation}%
is a standard Brownian motion, and  $\left(\frac{1}{\delta _{t}}S_{t}\right)_{t \in [0,T]}$
is a martingale with dynamics 
\begin{equation*}
d\left( \frac{S_{t}}{\delta _{t}}\right) =(\sigma
_{t}-\xi _{t})\frac{S_{t}}{\delta _{t}}\,dW_{t}^{\widetilde{\mathbb{Q}}}.
\end{equation*}%
Direct calculations yield 
\begin{equation*}
\frac{d\widetilde{\mathbb{Q}}}{d\mathbb{Q}}=\delta _{T}.
\end{equation*}%
Alternatively, $H_{t}$ may be also represented as 
\begin{equation}
H_{t}=\frac{E_{\mathbb{Q}}[\delta _{T}\int_{t}^{T}\frac{1}{2}(\lambda
_{s}-\xi _{s})^{2}\,ds|\mathcal{F}_{t}]}{E_{\mathbb{Q}}[\delta _{T}|\mathcal{F}%
_{t}]}=E_{\mathbb{Q}}\left[ \frac{\delta _{T}}{\delta _{t}}\int_{t}^{T}\left.\frac{%
1}{2}(\lambda _{s}-\xi _{s})^{2}\,ds \right\vert\mathcal{F}_{t}\right] ,  \label{def:H3}
\end{equation}%
which is obtained by using that 
\begin{equation*}
\frac{d\widetilde{\mathbb{Q}}}{d\mathbb{Q}}=\exp \left( -\frac{1}{2}%
\int_{0}^{T}\xi _{s}^{2}\,ds+\int_{0}^{T}\xi _{s}\,dW_{s}^{\mathbb{Q}}\right) .
\end{equation*}%
Finally, we introduce the processes $\left( M_{t}\right) _{t\in \left[ 0,T%
\right] }$ and $\left( \eta _{t}\right) _{t\in \left[ 0,T\right] }$ with 
\begin{equation}
M_{t}=\mathbb{E}_{\mathbb{\widetilde{Q}}}\left[ \frac{1}{2}\int_{0}^{T}\left.
\left( \lambda _{s}-\xi _{s}\right) ^{2}\,ds\right \vert \mathcal{F}_{t}\right] 
\text{ \  \  \  \ and \  \  \ }dM_{t}=\eta _{t}\,dW_{t}^{\widetilde{\mathbb{Q}}}.
\label{M-eta}
\end{equation}%
We are now ready to present the main result.

\begin{prop}\label{prop:single}
The following assertions hold:

\begin{enumerate}
\item The value function of \eqref{single-risk tolerance} is given by 
\begin{equation*}
v_{t}(x) =-\exp \left( -\frac{x}{\delta _{t}}-H_{t}\right) ,
\end{equation*}%
with $\delta $ and $H$ as in \eqref{delta-process} and \eqref{def:H2}.

\item The optimal strategy $\left( \pi _{s}^{\ast }\right) _{s\in \left[ t,T%
	\right] }$ is given by 
\begin{equation}
\pi _{s}^{\ast }=\delta _{s}\frac{\lambda _{s}-\eta _{s}-\xi _{s}}{\sigma
	_{s}}+\frac{\xi _{s}}{\sigma _{s}}x_{s}^{\ast },  \label{def:piast}
\end{equation}%
with $\xi ,\eta $ as in \eqref{xi-delta} and \eqref{M-eta}, and $x^{\ast }$
solving \eqref{X-i-complete} with $\pi ^{\ast }$ being used.

\item The optimal wealth $\left( x_{s}^{\ast }\right) _{s\in \left[t,T\right]
}$ solves 
\begin{equation*}
dx_{s}^{\ast }=\lambda _{s}\left( \delta _{s}(\lambda _{s}-\eta _{s}-\xi
_{s})+\xi _{s}x_{s}^{\ast }\right) \,ds+\left( \delta _{s}(\lambda _{s}-\eta
_{s}-\xi _{s})+\xi _{s}x_{s}^{\ast }\right)\, dW_{s}, \quad x_t^\ast = x, 
\end{equation*}%
and is given by 
\begin{equation}
x_{s}^{\ast }=x\Phi _{t,s} +\int_{t}^{s}\delta _{u}(\lambda _{u}-\xi
_{u})(\lambda _{u}-\eta _{u}-\xi _{u})\Phi _{u,s}\,du+\int_{t}^{s}\delta
_{u}(\lambda _{u}-\eta _{u}-\xi _{u})\Phi _{u,s}\,dW_{u},  \label{eq:Xt}
\end{equation} 
where, for $0 \leq u \leq s \leq T$,
\begin{equation*}
\Phi _{u,s}:=\exp \left( \int_{u}^{s}\left( \lambda _{v}-\frac{1}{2}\xi
_{v}\right) \xi _{v}\,dv+\int_{u}^{s}\xi _{v}\,dW_{v}\right) .
\end{equation*}

\end{enumerate}
\end{prop}

Using \eqref{eq:Xt}, \eqref{def:piast} gives the explicit representation of
the optimal policy, 
\begin{equation*}
\pi _{s}^{\ast }=\delta _{s}\frac{\lambda _{s}-\eta _{s}-\xi _{s}}{\sigma
	_{s}}+\frac{\xi _{s}}{\sigma _{s}}\left( x\Phi _{t,s}+\int_{t}^{s}\delta
_{u}(\lambda _{u}-\xi _{u})(\lambda _{u}-\eta _{u}-\xi _{u})\Phi
_{u,s}\,du+\int_{t}^{s}\delta _{u}(\lambda _{u}-\eta _{u}-\xi _{u})\Phi
_{u,s}\,dW_{u}\right) .
\end{equation*}

\subsubsection{The Markovian case}

We assume that the stock price process $\left( S_{t}\right) _{t\in \left[ 0,T%
\right] }$ solves 
\begin{equation*}
dS_{t}=\mu (t,S_{t})S_{t}\,dt+\sigma (t,S_{t})S_{t} \,dW_{t}\text{,}
\end{equation*}%
with the initial price $S_{0}>0,$ and the functions $\mu (t,S_{t})$ and $\sigma (t,S_{t})$
satisfying appropriate conditions, similar to the ones in Subsection~\ref{sec_markov} and
Assumption~\ref{assump:incomplete}. The risk tolerance is assumed to have the functional
representation 
\begin{equation*}
\delta _{T}=\delta (S_{T}),
\end{equation*}%
for some function $\delta :\mathbb{R}^{+}\rightarrow \mathbb{R}^{+}$ bounded
from below and such that $E_{\mathbb{P}}\left[ \delta ^{2}(S_{T})\right]
<\infty$, ({\it cf.} Assumption~\ref{assump:complete}).

The value function in \eqref{single-risk tolerance} takes the form
\begin{equation*}
V(t,x,S)=\sup_{\pi \in \mathcal{A}}E_{\mathbb{P}}\left[ -e^{-\frac{1}{\delta
\left( S_{T}\right) }x_{T}}\Big\vert x_{t}=x,S_{t}=S\right] ,
\end{equation*}%
and, in turn, Proposition~\ref{prop:single} yields
\begin{equation*}
V(t,x,S)=-\exp \left( \frac{x}{\delta (t,S)}-H(t,S)\right) ,
\end{equation*}%
with $\delta (t,S)$ and $H(t,S)$ solving%
\begin{equation*}
\delta _{t}+\frac{1}{2}\sigma ^{2}(t,S)S^{2}\delta _{SS}=0,\quad \delta
(T,S)=\delta (S),
\end{equation*}%
and 
\begin{equation*}
H_{t}+\frac{1}{2}\sigma ^{2}(t,S)S^{2}H_{SS}+\frac{1}{\delta(t,S) }\sigma ^{2}(t,S)S^{2}\delta _{S}(t, S) H_{S}+\frac{1}{2}\left( \lambda (t,S)-\frac{1}{\delta(t,S)}\sigma (t,S)S\delta _{S}(t, S)\right)^{2}=0,\quad H(T,S)=0.
\end{equation*}%
Clearly, 
\begin{equation*}
\delta (t,S)=E_{\mathbb{Q}}\left[ \left. \delta (S_{T})\right \vert S_{t}=S%
\right],
\end{equation*}%
and 
\begin{equation*}
H(t,S)=E_{\mathbb{\widetilde{Q}}}\left[ \int_{t}^{T}\frac{1}{2}\left(
\lambda (u,S_{u})-\sigma (u,S_{u})S_{u}\frac{\delta _{S}(u,S_{u})}{\delta
(u,S_{u})}\right) ^{2}du\bigg \vert S_{t}=S\right] ,
\end{equation*}%
and, furthermore, 
\begin{equation*}
\xi _{t}=\frac{\delta _{S}(t,S_{t})}{\delta (t,S_{t})}S_{t}\sigma (t,S_{t})%
\text{ \  \  \  \ and }\quad \eta _{t}=H_{S}(t,S_{t})S_{t}\sigma (t,S_{t}).
\end{equation*}%
Using the above relations and \eqref{def:piast}, we derive the optimal investment
process, 
\begin{equation*}
\pi _{s}^{\ast }=\delta (s,S_{s})\left( \frac{\lambda (s,S_{s})}{\sigma
(s,S_{s})}-S_{s}H_{S}(s,S_{s})\right) +\delta _{S}(s,S_{s})S_{s}\left( -1+%
\frac{1}{\delta \left( s,S_{s}\right) }x_{s}^{\ast }\right).
\end{equation*}%

For completeness, we note that if $\delta _{T} \equiv \delta >0$,  the above
expression simplify to (see \cite{SiZa:05}) 
\begin{equation*}
V(t,x,S)=-e^{-\frac{1}{\delta }x-H(t,S)},
\end{equation*}%
with $H(t,S)$ solving 
\begin{equation*}
H_{t}+\frac{1}{2}\sigma ^{2}(t,S)S^{2}H_{SS}+\frac{1}{2}\lambda
^{2}(t,S)=0,\quad H(T,S)=0.
\end{equation*}%
The optimal strategy reduces to 
\begin{equation*}
\pi _{s}^{\ast }=\delta \left( \frac{\lambda (s,S_{s})}{\sigma (s,S_{s})}%
-S_{s}H_{S}(s,S_{s})\right) .
\end{equation*}

\subsection{$N$-player game}

We now study the $N$-player game. The concepts and various quantities are in
direct analogy to those in Section~\ref{sec_Nagent} and, thus, we omit various
intermediate steps and only focus on the new elements coming from the
randomness of the risk tolerance coefficients.

\begin{prop}
For $i\in \mathcal{I}$, let 
\begin{equation*}
\delta _{t}^{i}=E_{\mathbb{Q}}[ \left. \delta _{T}^{i}\right \vert 
\mathcal{F}_{t}] ,
\end{equation*}%
with $\mathbb{Q}$ as in \eqref{Q-measure} and $\left( \xi _{t}^{i}\right)
_{t\in \left[ 0,T\right] }$ be such that 
\begin{equation*}
d\delta _{t}^{i}=\xi _{t}^{i}\delta _{t}^{i}\,dW_{t}^{\mathbb{Q}}.
\end{equation*}
Define the measure $\widetilde Q^i$ on $\mc{F}_T$ as 
\begin{equation}
\frac{d\widetilde{\mathbb{Q}}^i}{d\mathbb{P}}=\exp \left( -\frac{1}{2}%
\int_{0}^{T}(\lambda _{s}-\xi _{s}^i)^{2}\,ds-\int_{0}^{T}\left( \lambda
_{s}-\xi _{s}^i\right) \,dW_{s}\right),
\end{equation}
and the processes $(M_t^i)_{t \in [0,T]}$ and $(\eta_t)_{t \in [0,T]}$ with
\begin{equation}
M_{t}^i=\mathbb{E}_{\mathbb{\widetilde{Q}}^i}\left[\frac{1}{2}\int_{0}^{T}\left.
\left( \lambda _{s}-\xi _{s}^i\right) ^{2}\,ds\right \vert \mathcal{F}_{t}\right] 
\text{ \  \  \  \ and \  \  \ }dM_{t}^i=\eta _{t}^i\,dW_{t}^{\widetilde{\mathbb{Q}}^i}.
\end{equation}
Let also, 
\begin{equation*}
\psi _{N}=\frac{1}{N}\sum_{i=1}^{N} c_i,
\end{equation*}%
and assume that $\psi _{N}<1$. Then 

\begin{enumerate}

\item  The player $i$'s game value \eqref{V-i-risktolerance} is given by 
\begin{equation*} 
V^{i}(x_{1}, \ldots, x_{i}, \ldots, x_{N} ) =-\exp \left( -\frac{1}{E_{%
\mathbb{Q}}[\delta _{T}^{i}]\text{ }}\bigl(x_{i}-\frac{c_{i}}{N}\Sigma_{j=1
}^N x_{j}\bigr)-E_{\widetilde{\mathbb{Q}}^{i}}\left[ \frac{1}{2}\int_{0}^{T}\left(
\lambda _{s}-\xi _{s}^{i}\right) ^{2}\,ds\right] \right).
\end{equation*}%

\item The equilibrium strategies $(\pi _{t}^{1,\ast },\ldots ,\pi _{t}^{N,\ast
})_{t\in \left[ 0,T\right] }$ are given by 
\begin{equation}
\pi _{t}^{i,\ast } = c_i \bar \pi_t^\ast + \frac{1%
}{\sigma _{t}}\left( \delta _{t}^{i}(\lambda _{t}-\xi _{t}^{i}-\eta
_{t}^{i})+\bigl(X_{t}^{i,\ast }-\frac{c_{i}}{N}\sum_{j = 1}^N X_{t}^{j,\ast
}\bigr)\xi _{t}^{i}\right) ,  \label{eq:pirelation}
\end{equation}%
where $\bar\pi_t^\ast := \frac{1}{N} \Sigma_{j=1}^N \pi_t^{j, \ast}$ is defined as
\begin{equation}\label{eq:pibar-complete}
\bar\pi_t^\ast  = \frac{1}{1-\psi _{N}}\frac{1}{\sigma _{t}}\left(
\lambda _{t}\varphi _{N}^{1}(t)-\varphi _{N}^{2}(t)+\varphi
_{N}^{3}(t)-\varphi _{N}^{4}(t)\bar{X}_{t}^{\ast }\right),
\end{equation}
with
\begin{align*}
\varphi _{N}^{1}(t)&=\frac{1}{N}\Sigma _{j=1}^{N}\delta _{t}^{j}, \quad \varphi _{N}^{2}(t)=\frac{1}{N}\Sigma
_{j=1}^{N}\delta _{t}^{j}(\xi _{t}^{j}+\eta
_{t}^{j}), \\
\varphi _{N}^{3}(t)&=\frac{1}{N}\Sigma _{j=1}^{N}X_{t}^{j,\ast }\xi _{t}^{j}, \quad \varphi _{N}^{4}(t)=\Sigma _{j=1}^{N}c_j\xi _{t}^{j}.
\end{align*}%

\item The associated optimal wealth processes $\left( X_{t}^{i,\ast }\right)
_{t\in \left[ 0,T\right] }$ are given by%
\begin{equation}
X_{t}^{i,\ast }=c_i\bar{X}%
_{t}^{\ast } + \left( \tilde{x}_{i}\Phi
_{0,t}^{i}+\int_{0}^{t}(\lambda _{s}-\xi _{s}^{i})\delta _{s}^{i}(\lambda
_{s}-\eta _{s}^{i}-\xi _{s}^{i})\Phi _{s,t}^{i}\,ds+\int_{0}^{t}\delta
_{s}^{i}(\lambda _{s}-\eta _{s}^{i}-\xi _{s}^{i})\Phi
_{s,t}^{i}\,dW_{s}\right),\label{eq-Xt-complete}
\end{equation}%
with 
\begin{equation*}
\bar{X}_{t}^{\ast }:=\frac{1}{1-\psi _{N}}\left( \frac{1}{N}\Sigma
_{i=1}^{N}\left( \tilde{x}_{i}\Phi
_{0,t}^{i}+\int_{0}^{t}\delta _{s}^{i}(\lambda _{s}-\xi _{s}^{i})(\lambda
_{s}-\eta _{s}^{i}-\xi _{s}^{i})\Phi _{s,t}^{i}\,ds+\int_{0}^{t}\delta
_{s}^{i}(\lambda _{s}-\eta _{s}^{i}-\xi _{s}^{i})\Phi
_{s,t}^{i}\,dW_{s}\right) \right),
\end{equation*}%
where $ \tilde{x}_{i}=x_{i}-\frac{c_{i}}{N}\Sigma _{j=1}^{N}x_{j}$, and 
\begin{equation}
\Phi _{s,t}^{i}:=\exp \left( \int_{s}^{t}\left( \lambda _{u}-\frac{1}{2}\xi
_{u}^{i}\right) \xi _{u}^{i}\,du+\int_{s}^{t}\xi _{u}^{i}\,dW_{u}\right) \text{.}
\label{Phis}
\end{equation}

\end{enumerate}
\end{prop}

\begin{proof}
Using the dynamics of $X^{1},\ldots ,X^{N}$ in \eqref{X-i-complete}, problem \eqref{V-i-risktolerance} reduces to 
\begin{equation*}
v\left( \tilde{x}\right) =\sup_{\widetilde{\pi }^{i}\in \mathcal{A}}E_{%
\mathbb{P}}\left[ -\exp \left( -\frac{1}{\delta _{T}^{i}}\widetilde{X}%
_{T}^{i}\right) \right] ,
\end{equation*}%
where $\widetilde{X}_{t}^{i}=X_{t}^{i}-\frac{c_{i}}{N}\Sigma
_{j=1}^{N}X_{t}^{j}$ satisfies $d\widetilde{X}_{t}^{i}=\widetilde{\pi }%
_{t}^{i}\left( \mu _{t}\,dt+\sigma _{t}\,dW_{t}\right) $ with $\widetilde{X}%
_{0}^{i}=\tilde{x}_{i}.$ Taking $\pi^j \in \mc{A}$, $j \neq i$, as fixed and using Proposition~\ref{prop:single}, we deduce that
$\pi^{i, \ast}$ satisfies
\begin{equation}\label{eq:pistar}
\widetilde{\pi }_{t}^{i,\ast }=\pi _{t}^{i,\ast }-\frac{c_{i}}{N}\left(
\Sigma _{j\neq i}\pi _{t}^{j}+\pi _{t}^{i,\ast }\right) = \delta _{t}^i \frac{\lambda _{t}-\eta _{t}^i-\xi _{t}^i}{\sigma
	_{t}}+\frac{\xi _{t}^i}{\sigma _{t}}\widetilde X_{t}^{i, \ast },
\end{equation}%
where $\widetilde X_{t}^{i, \ast}$ is the wealth process $\widetilde X_{t}^{i}$ associated with the strategy $\widetilde{\pi }_{t}^{i,\ast }$.

At equilibrium, $\pi_t^j$ in \eqref{eq:pistar} coincides with $\pi_t^{j, \ast}$. Therefore, averaging over $i\in \mathcal{I}$ gives%
\begin{equation*}
\bar{\pi}_{t}^{\ast }-\psi_N\bar{\pi}_{t}^{\ast } = \frac{1}{\sigma _{t}}\left( \lambda _{t}\varphi_N^1(t)-\varphi_N^2(t)+\varphi_N^3(t)-\varphi_N^4(t)\bar{X}_{t}^{\ast }\right) .
\end{equation*}
Dividing both sides by $1-\psi_N$ yields \eqref{eq:pibar-complete}, and then \eqref{eq:pirelation} follows.

To obtain explicit expressions of $X_t^{i, \ast}$ and $\bar{X}_{t}^{\ast }$, we solve for $%
\widetilde{X}_{t}^{i,\ast }$ using the optimal
strategy deduced in Section~\ref{sec:single-complete} ({\it cf.} \eqref{def:piast}). We then obtain
\begin{equation*}
\widetilde{X}_{t}^{i,\ast }=X_{t}^{i,\ast }-\frac{c_{i}}{N}\sum_{j=1
}^N X_{t}^{j,\ast } =\tilde{x}_{i}\Phi _{0,t}^{i}+\int_{0}^{t}\delta _{s}^{i}(\lambda _{s}-\xi
_{s}^{i})(\lambda _{s}-\eta _{s}^{i}-\xi _{s}^{i})\Phi
_{s,t}^{i}\,ds+\int_{0}^{t}\delta _{s}^{i}(\lambda _{s}-\eta _{s}^{i}-\xi
_{s}^{i})\Phi _{s,t}^{i}\,dW_{s},
\end{equation*}
with $\Phi _{s,t}^{i}$ as in \eqref{Phis}. We conclude by averaging over all $i \in \mathcal{I}$.
\end{proof}

\subsection{The It\^{o}-diffusion common-noise MFG}

Let $(\Omega ,\mathcal{F},\mathbb{P})$ be a probability space that
supports the Brownian motion $W$ as well as the {\it random} {type} vector 
\begin{equation*}
\theta =(x,\delta _{T},c),
\end{equation*}%
which is independent of $W$. As before, we denote by $(\mathcal{F}%
_{t})_{t\in \lbrack 0,T]}$ the natural filtration generated by $W$, and $(\mathcal{F}_{t}^{{MF}})_{t\in \lbrack 0,T]}$
with $\mathcal{F}_{t}^{{MF}}=\mathcal{F}_{t}\vee \sigma (\theta )$. In
the mean-field setting, we model the representative player. One may also
think of a continuum of players whose initial wealth $x$ and the interaction
parameter $c$ are random, chosen at initial time $0$, similar to the MFG in Section~\ref{sec_MFG-incomplete} herein. However, now, their risk tolerance
coefficients have \emph{two} sources of randomness, related to their form and their
terminal (at $T$) measurability, respectively. Specifically, at initial time $0$, it is
determined how these coefficients will depend on the final information, provided at $T$. For example, in the
Markovian case, this amounts to (randomly) selecting at time $0$ the functional form
of $\delta (\cdot )$ and, in turn, the risk tolerance used for utility
maximization is given by the random variable $\delta (S_{T})$, which depends
on the information $\mathcal{F}_{T}$ through $S_{T}.$

Similarly to \eqref{V-i-risktolerance}, we are concerned with the optimization problem 
\begin{equation}
V(x) =\sup_{\pi \in \mathcal{A}^{MF}}E_{\mathbb{P}}\left[ \left.
-\exp \left( -\frac{1}{\delta _{T}}\left( X_{T}^{\pi }-c\overline{X}\right)
\right) \right \vert \mathcal{F}_{0}^{MF},\text{ }X_{0}=x\right] ,
\label{def:MFG}
\end{equation}%
and the definition of the mean-field game is analogous to Definition~\ref{def:MFG-incomplete}.

Let the processes $\left( \delta _{t}\right) _{t\in \left[ 0,T\right] }$ and $\left( \xi
_{t}\right) _{t\in \left[ 0,T\right] }$ be given by%
\begin{equation}
\delta _{t}=E_{\mathbb{Q}}[\delta _{T}|\mathcal{F}_{t}^{MF}]\text{ \  \  \ and
\  \ }d\delta _{t}=\xi _{t}\delta _{t}\,dW_{t}^{\mathbb{Q}},  \label{delta-xi}
\end{equation}%
with $\mathbb{Q}$ defined on $\mathcal{F}_{T}^{MF}$ by \eqref{Q-measure}. The process $(\delta_t)_{t \in [0, T]}$  may be interpreted as the arbitrage-free price of the risk tolerance ``claim'' $\delta_T$ for this \emph{representative} player.
Let also $\mathbb{\widetilde{Q}}$ be defined on $\mathcal{F}_{T}^{MF}$ by 
\begin{equation*}
\frac{d\widetilde{\mathbb{Q}}}{d\mathbb{Q}}=\delta _{T},
\end{equation*}%
and consider the martingale $M_{t}=E_{\mathbb{\widetilde{Q}}}\left[ \left. \frac{1}{2%
}\int_{0}^{T}\left( \lambda _{s}-\xi _{s}\right) ^{2}ds\right \vert \mathcal{%
F}_{t}^{MF}\right] $ and $\left( \eta _{t}\right) _{t\in \left[ 0,T\right] }$
to be such that 
\begin{equation}
dM_{t}=\eta _{t}\,dW_{t}^{\widetilde{\mathbb{Q}}},  \label{eta}
\end{equation}%
with $W_{t}^{\widetilde{\mathbb{Q}}}=W_{t}+\int_{0}^{t}\left( \lambda
_{s}-\xi _{s}\right) ds.$ The processes $\delta ,\xi $ and $\eta $ are all $%
\mathcal{F}_{t}^{MF}$-adapted.

\medskip
We now state the main result of this section.

\begin{prop}\label{prop:mfg-complete}
If $E_{\mathbb{P}}[c] <1$, there exists a
MFG equilibrium $\left( \pi _{t}^{\ast }\right) _{t\in \left[ 0,T\right] }$,
given by%
\begin{multline}\label{eq_piast-complete}
\pi _{t}^{\ast }=\frac{c}{1-E_{\mathbb{P}}[c]}\frac{1}{\sigma _{t}}\left(
\lambda _{t}E_{\mathbb{Q}}[\delta _{T}|\mathcal{F}_{t}]-E_{\mathbb{Q}%
}[\delta _{T}(\xi _{t}+\eta _{t})|\mathcal{F}_{t}]+E_{\mathbb{P}%
}[X_{t}^{\ast }\xi _{t}|\mathcal{F}_{t}]-E_{\mathbb{P}}[c\xi _{t}|\mathcal{F}%
_{t}]E_{\mathbb{P}}[X_{t}^{\ast }|\mathcal{F}_{t}]\right)\\
+\frac{1}{\sigma _{t}}\left(\delta _{t}(\lambda _{t}-\xi _{t}-\eta
_{t})+(X_{t}^{\ast }-cE_{\mathbb{P}}[X_{t}^{\ast }|\mathcal{F}_{t}])\xi
_{t}\right),
\end{multline}%
with $\delta, \xi $ and $\eta $ as in \eqref{delta-xi} and \eqref{eta}, and 
$\left( X_{t}^{\ast }\right) _{t\in \left[ 0,T\right] }$ being the
associated optimal wealth process, solving
\begin{equation}\label{eq:Xoptimal-complete}
\,d X_t^\ast = \pi_t^\ast (\mu_t \, dt + \sigma_t \,dW_t).
\end{equation}

The value of the MFG is given by 
\begin{equation*}
V(x)=-\exp \left( -\frac{1}{E_{\mathbb{Q}}[\delta _{T} \vert \mc{F}_0^{MF}]\text{ }}(x-cm)-E_{%
\widetilde{\mathbb{Q}}}\left[ \frac{1}{2}\int_{0}^{T}\left( \lambda _{s}-\xi
_{s}\right) ^{2}ds \Big\vert \mc{F}_0^{MF}\right] \right), \ m = E_{\mathbb{P}}[x].
\end{equation*}
\end{prop}

For the proof, we will need the following lemma.

\begin{lem}\label{lem:condE}
If $X$ is a $\mathcal{F}_{s}^{MF}$-measurable integrable
random variable, then $E_{\mathbb{P}}[X|\mathcal{F}_{t}]=E_{\mathbb{P}}[X|%
\mathcal{F}_{s}]$, for $s\in \left[ 0,t\right] $. 
\end{lem}

\begin{proof}
Let $\mathcal{P}:=\{A=C\cap D:C\in \mathcal{F}_{s},\;D\in
\sigma \{W_{u}-W_{s},s\leq u\leq t\} \}$ and $\mathcal{L}=\{A\in \mathcal{F}%
:E_{\mathbb{P}}[X\mathbf{1}_{A}]=E_{\mathbb{P}}[E_{\mathbb{P}}[X|\mathcal{F}%
_{s}]\mathbf{1}_{A}]\}$. Then, the following assertions hold:

(1) $\mathcal{P}$ is a $\pi$-system since both $\mathcal{F}_s$ and $\sigma
\{W_u - W_s, s \leq u \leq t\}$ are $\sigma$-algebras and closed under
intersection. Also $\mathcal{F}_s \subseteq \mathcal{P}$ and $\sigma \{W_u -
W_s, s \leq u \leq t\} \subseteq \mathcal{P}$ by taking $D = \Omega$ and $C
= \Omega$.

(2) $\mathcal{P}\subseteq \mathcal{L}$. For any $A\in \mathcal{P}$, $A=C\cap
D$ with $C\in \mathcal{F}_{s},\;D\in \sigma \{W_{u}-W_{s},s\leq u\leq t\}$,
it holds that
\begin{equation*}
E_{\mathbb{P}}[E_{\mathbb{P}}[X|\mathcal{F}_{s}]\mathbf{1}_{A}]=E_{\mathbb{P}%
}[E_{\mathbb{P}}[X|\mathcal{F}_{s}]\mathbf{1}_{C}\mathbf{1}_{D}]=E_{\mathbb{P%
}}[E_{\mathbb{P}}[X\mathbf{1}_{C}|\mathcal{F}_{s}]\mathbf{1}_{D}]=E_{\mathbb{%
P}}[X\mathbf{1}_{C}]E_{\mathbb{P}}[\mathbf{1}_{D}],
\end{equation*}%
where we have consecutively used that $C\perp D$, the metastability of $%
\mathbf{1}_{C}$, and the independence between $\mathbf{1}_{D}$ and $\mathcal{%
F}_{s}$.

Furthermore, by the independence between $\mathbf{1}_{D}$ and $\mathcal{F}_{s}^{{MF}%
}=\mathcal{F}_{t}\vee \sigma (\theta )$, we deduce 
\begin{equation*}
E_{\mathbb{P}}[X\mathbf{1}_{A}]=E_{\mathbb{P}}[X\mathbf{1}_{C}\mathbf{1}%
_{D}]=E_{\mathbb{P}}[X\mathbf{1}_{C}]E_{\mathbb{P}}[\mathbf{1}_{D}],
\end{equation*}%
and conclude that $A\in \mathcal{L}$. Therefore $\mathcal{P}\subseteq 
\mathcal{L}$.

(3) $\mathcal{L}$ is a $\lambda $-system. It is obvious that $\Omega \in 
\mathcal{L}$ and $A\in \mathcal{L}$ imply that $A^{c}\in \mathcal{L}$. For a
sequence of disjoint sets $A_{1},A_{2},\ldots $ in $\mathcal{L}$, one has $%
\left \vert X\mathbf{1}_{\cup _{i=1}^{\infty }A_{i}}\right \vert \leq
\left
\vert X\right \vert $ and, thus, by the dominated convergence
theorem, we deduce that 
\begin{equation}
E_{\mathbb{P}}[X\mathbf{1}_{\cup _{i=1}^{\infty }A_{i}}]=\sum_{i=1}^{\infty }E_{\mathbb{P}}[X\mathbf{1}_{A_{i}}].  \label{eq:cond1}
\end{equation}%
Similarly, by the inequalities $\Vert E_{\mathbb{P}}[X|\mathcal{F}_{s}]\mathbf{1}_{\cup
_{i=1}^{\infty }A_{i}}\Vert _{1}\leq \Vert E_{\mathbb{P}}[X|\mathcal{F}%
_{s}]\Vert _{1}\leq \Vert X\Vert _{1}$, we have
\begin{equation}
E_{\mathbb{P}}[E_{\mathbb{P}}[X|\mathcal{F}_{s}]\mathbf{1}_{\cup
_{i=1}^{\infty }A_{i}}]=\sum_{i=1}^{\infty }E_{\mathbb{P}}[E_{\mathbb{P}}[X|%
\mathcal{F}_{s}]\mathbf{1}_{A_{i}}].  \label{eq:cond2}
\end{equation}%
Since $A_{i}\in \mathcal{L}$, $\forall i$, the right-hand-sides of
\eqref{eq:cond1} and \eqref{eq:cond2} are equal, which implies
$\cup _{i=1}^{\infty }A_{i}\in \mathcal{L}$.

Therefore, by the $\pi $-$\lambda $ theorem, we obtain that $\mathcal{F}_{t}=\sigma (%
\mathcal{F}_{s}\cup \sigma \{W_{u}-W_{s},s\leq u\leq t\})\subseteq \sigma (%
\mathcal{P})\subseteq \mathcal{L}$. Noticing that $E_{\mathbb{P}}[X|\mathcal{%
F}_{s}]$ is $\mathcal{F}_{t}$-measurable by definition, we have that $E_{%
\mathbb{P}}[X|\mathcal{F}_{t}]=E_{\mathbb{P}}[X|\mathcal{F}_{s}]$. 
\end{proof}

\begin{proof}[Proof of Proposition~\ref{prop:mfg-complete}]
Let $\left( X_{t}^{\alpha }\right) _{t\in \left[ 0,T\right]
} $ be given by $X_{t}^{\alpha }=x+\int_{0}^{t}\mu _{s}\alpha
_{s}\,ds+\int_{0}^{t}\sigma _{s}\alpha _{s}\,dW_{s}$ for an admissible policy $%
\alpha _{t}$ ($\mathcal{F}_{t}^{MF}$-adapted) and define $\overline{X}%
_{t}:=E_{\mathbb{P}}[X_{t}^{\alpha }|\mathcal{F}_{t}].$ Then, 
\begin{equation*}
\overline{X}_{t}=m+E_{\mathbb{P}}\left[ \int_{0}^{t}\mu _{s}\alpha _{s}\,ds\Big|%
\mathcal{F}_{s}\right] +E_{\mathbb{P}}\left[ \int_{0}^{t}\sigma _{s}\alpha
_{s}\,dW_{s}\Big|\mathcal{F}_{s}\right] .
\end{equation*}%
Using Lemma~\ref{lem:condE}, the adaptivity of $\mu _{t}$, $\sigma _{t}$
with respect to $\mathcal{F}_{t}$, and the definition of It\^{o} integral,
we rewrite the above as 
\begin{equation*}
\overline{X}_{t}=m+\int_{0}^{t}\mu _{s}E_{\mathbb{P}}\left[ \alpha _{s}|%
\mathcal{F}_{s}\right] \,ds+\int_{0}^{t}\sigma _{s}E_{\mathbb{P}}\left[ \alpha
_{s}|\mathcal{F}_{s}\right] \,dW_{s}.
\end{equation*}%
Direct arguments yield that the optimization problem \eqref{def:MFG} reduces
to%
\begin{equation*}
V\left( \tilde{x}\right) =\sup_{\widetilde{\pi}\in \mathcal{A}^{MF}}E_{\mathbb{P}}%
\left[ -\exp \left( -\frac{1}{\delta _{T}}\widetilde{X}_{T}\right) \Big \vert \mc{F}_0^{MF}, \widetilde X_0 = \tilde x \right] ,
\end{equation*}%
where $( \widetilde{X}_t) _{t\in \left[ 0,T\right] }$ solves
\begin{equation}
d\widetilde{X}_{t}\equiv d(X_{t}-c\overline{X}_{t})=\widetilde{\pi }_{t}(\mu
_{t}\,dt+\sigma _{t}\,dW_{t}),  \label{X-telda}
\end{equation}%
with $\widetilde{X}_{0}=\tilde{x}=x-cm$ and $\widetilde{\pi }_{t}:=\pi _{t}-cE_{%
\mathbb{P}}[\alpha _{t}|\mathcal{F}_{t}]$. Then, \eqref{def:piast}
yields
\begin{equation}
\widetilde{\pi }_{t}^{\ast }=\delta _{t}\frac{\lambda _{t}-\eta _{t}-\xi _{t}%
}{\sigma _{t}}+\frac{\xi _{t}}{\sigma _{t}}\widetilde{X}_{t}^{\ast },
\label{eq:tildepiast}
\end{equation}%
with $\delta _{t},\xi _{t},\eta _{t}$ given in \eqref{delta-xi} and \eqref{eta}, and $( \widetilde{X}_{t}^{\ast }) _{t\in \left[ 0,T\right] }$
solving \eqref{X-telda} with $\widetilde{\pi }^{\ast }$ being used. On the
other hand, using that $\widetilde{\pi }_{t}^{\ast }=\pi _{t}^{\ast }-cE_{%
\mathbb{P}}[\alpha _{t}|\mathcal{F}_{t}],$ we obtain
\begin{equation*}
\pi _{t}^{\ast }-cE_{\mathbb{P}}[\alpha _{t}|\mathcal{F}_{t}]=\delta _{t}%
\frac{\lambda _{t}-\eta _{t}-\xi _{t}}{\sigma _{t}}+\frac{\xi _{t}}{\sigma
_{t}}\widetilde{X}_{t}^{\ast }.
\end{equation*}%
In turn, using that, at equilibrium, $\alpha =$ $\pi ^{\ast }$, we get 
\begin{equation*}
(1-E_{\mathbb{P}}[c])E_{\mathbb{P}}[\pi _{t}^{\ast }|\mathcal{F}_{t}]=\frac{1%
}{\sigma _{t}}\left( \lambda _{t}E_{\mathbb{P}}[\delta _{t}|\mathcal{F}%
_{t}]-E_{\mathbb{P}}[\delta _{t}(\xi _{t}+\eta _{t})|\mathcal{F}_{t}]+E_{%
\mathbb{P}}[\widetilde{X}_{t}^{\ast }\xi _{t}|\mathcal{F}_{t}]\right) .
\end{equation*}%
Further calculations give
\begin{multline}
\pi _{t}^{\ast }=c\frac{1}{1-E_{\mathbb{P}}[c]}\frac{1}{\sigma _{t}}\left(
\lambda _{t}E_{\mathbb{P}}[\delta _{t}|\mathcal{F}_{t}]-E_{\mathbb{P}%
}[\delta _{t}(\xi _{t}+\eta _{t})|\mathcal{F}_{t}]+E_{\mathbb{P}%
}[X_{t}^{\ast }\xi _{t}|\mathcal{F}_{t}]-E_{\mathbb{P}}[X_{t}^{\ast }|%
\mathcal{F}_{t}]E_{\mathbb{P}}[c\xi _{t}|\mathcal{F}_{t}]\right)\\
+\frac{\delta _{t}(\lambda _{t}-\eta _{t})-\delta _{t}\xi _{t}+X_{t}^{\ast
}\xi _{t}-c\xi _{t}E_{\mathbb{P}}[X_{t}^{\ast }|\mathcal{F}_{t}]}{\sigma _{t}%
}.
\end{multline}%
Finally, we obtain
\begin{equation*}
E_{\mathbb{P}}[\delta _{t}|\mathcal{F}_{t}]=E_{\mathbb{P}}[E_{\mathbb{Q}%
}[\delta _{T}|\mathcal{F}_{t}^{MF}]|\mathcal{F}_{t}]=E_{\mathbb{P}}\left[E_{%
\mathbb{P}}\left[\frac{\delta _{T}Z_{T}}{Z_{t}}\Big|\mathcal{F}_{t}^{MF}\right]\Big|\mathcal{F}%
_{t}\right]=E_{\mathbb{P}}\left[\frac{\delta _{T}Z_{T}}{Z_{t}}\Big|\mathcal{F}_{t}\right]=E_{%
\mathbb{Q}}[\delta _{T}|\mathcal{F}_{t}],
\end{equation*}%
and a similar derivation for $E_{\mathbb{P}}[\left. \delta _{t}(\xi
_{t}+\eta _{t})\right \vert \mathcal{F}_{t}]$. We conclude by checking the admissibility of $\pi^\ast$ which follows from model assumptions, the form of $\pi^\ast$, and equation~\eqref{eq:Xoptimal-complete}. 
\end{proof}

\section{Conclusions and future research directions}\label{sec:conclude}

In It\^{o}-diffusion environments, we introduced and studied a family of $N$-player and common-noise mean-field games in the context of optimal portfolio choice in a common market. The players aim to maximize their expected terminal utility, which depends on their own wealth and the wealth of their peers.

We focused on two cases of exponential utilities, specifically, the classical CARA case and the extended CARA case with random risk tolerance. The former was considered for the incomplete market model while the latter for the complete one. We provided the equilibrium processes and the values of the games in explicit (incomplete market case) and in closed form (complete market case). We note that in the case of random risk tolerances, for which even the single-player case is interesting in its own right, the optimal strategy process depends on the state process, even if the preferences are of exponential type. 

A natural extension is to consider power utilities (CRRA), which are also commonly used in models of portfolio choice. This extension, however, is by no means straightforward. Firstly, in the incomplete market case, the underlying measure depends on the individual risk tolerance, which is not the case for the CARA utilities considered herein (see \eqref{def_MM} for the minimal martingale measure and \eqref{def_ME}-\eqref{ME_BSDE} for the minimal entropy measure, respectively). Secondly, while it is formally clear how to formulate the random risk tolerance case for power utilities, its solution is far from obvious. The authors are working in both these directions.

Our results may be used to study such models when the dynamics of the common market and/or the individual preferences are not entirely known. This could extend the analysis to various problems in reinforcement learning (see, for example, the recent work \cite{leng2020learning} in a static setting). It is expected that results similar to the ones in \cite{wang2020reinforcement} could be derived and, in turn, used to build suitable algorithms (see, also,  \cite{guo2020entropy} for a Markovian case). 

\section*{Acknowledgments}
RH was partially supported by the NSF grant DMS-1953035, and the Faculty Career Development Award and the Research Assistant Program Award at UCSB.

This work was presented at the SIAM Conference on Financial Mathematics and Engineering in 2021. The authors would like to thank the participants for fruitful comments and suggestions.

\bibliographystyle{plain}
\bibliography{Reference}

\end{document}